\documentclass[11pt]{article}

\usepackage{times,latexsym,graphicx,amsmath,amsfonts,amsthm}
\usepackage{natbib}

\newtheoremstyle{localthm}
	{5pt} 
	{5pt} 
	{\sl} 
	{} 
	{\bf} 
	{{\rm.}} 
	{.7em} 
	{} 

\theoremstyle{localthm}
\newtheorem{Theorem}{Theorem}
\newtheorem{Lemma}[Theorem]{Lemma}
\newtheorem{Corollary}[Theorem]{Corollary}
\newtheorem{Proposition}[Theorem]{Proposition}

\newtheoremstyle{localrem}
	{5pt} 
	{5pt} 
	{\rm} 
	{} 
	{\bf} 
	{{\rm.}} 
	{.7em} 
	{} 

\theoremstyle{localrem}

\def\bs{\boldsymbol}

\def\vb{\bs{v}}
\def\wb{\bs{w}}
\def\x{\bs{x}}
\def\y{\bs{y}}
\def\z{\bs{z}}

\def\V{\bs{V}}
\def\X{\bs{X}}
\def\Y{\bs{Y}}

\def\bsf{\bs{f}}
\def\bF{\bs{F}}
\def\bmu{\bs{\mu}}

\def\bTheta{\bs{\Theta}}
\def\bSigma{\bs{\Sigma}}
\def\bDelta{\bs{\Delta}}

\def\Vi{\V_{\!\!i}}
\def\Xi{\X_{\!i}}
\def\Yi{\Y_{\!\!i}}
\def\Yzero{\Y_{\!\!\bs{0}}}
\def\Yz{\Y_{\!\!\z}}

\def\Cov{\mathrm{Cov}}
\def\Var{\mathrm{Var}}
\def\vMF{\mathrm{vMF}}
\def\Bh{\mathrm{Bh}}

\def\Ex{\operatorname{\mathbb{E}}}
\def\Pr{\operatorname{\mathbb{P}}}

\def\trace{\operatorname{\mathrm{trace}}}
\def\argmin{\mathop{\rm arg\,min}}

\def\R{\mathbb{R}}
\def\S{\mathbb{S}}

\def\d{\mathrm{d}}

\def\FF{\mathcal{F}}
\def\LL{\mathcal{L}}
\def\XX{\mathcal{X}}

\def\hat{\widehat}

\begin{document}

\addtolength{\baselineskip}{+.15\baselineskip}

\title{Nonparametric Smoothing of Directional and Axial Data}

\author{Lutz D{\"u}mbgen$^1$ and Caroline Haslebacher$^{1,2}$\\[1ex]
\it{$^1$University of Bern}\\
\it{$^2$Southwest Research Institute, Boulder, Colorado, USA}}

\date{July 2025}

\maketitle

\paragraph{Abstract.}
We discuss generalized linear models for directional data where the conditional distribution of the response is a von Mises--Fisher distribution in arbitrary dimension or a Bingham distribution on the unit circle. To do this properly, we parametrize von Mises--Fisher distributions by Euclidean parameters and investigate computational aspects of this parametrization. Then we modify this approach for local polynomial regression as a means of nonparametric smoothing of distributional data. The methods are illustrated with simulated data and a data set from planetary sciences involving covariate vectors on a sphere with axial response.



\section{Introduction}
\label{sec:Introduction}

The starting point for the present paper are data sets from Planetary Sciences. The Galileo Mission yielded raw imaging data from the moons Ganymede and Europa of the planet Jupiter. Both moons have a global subsurface ocean with an ice shell on top. From pictures of linear surface features on the icy surface, the physicists extracted observation pairs $(\Xi,\Vi)$, $1 \le i \le n$, consisting of a center position $\Xi$ on the surface at which a feature, presumably a crack in the ice layer, could be identified, and the direction $\Vi$ of that crack. Identifying the surface of the moon with the unit sphere $\S^2$ of $\R^3$, we are talking about points $\Xi \in \S^2$ and $\Vi \in \S^2 \cap \Xi^\perp$, where $\Vi$ represents the axis $\R\Vi$. Figures~\ref{fig:Europa_raw} and \ref{fig:Europa_processed} illustrate this process for the moon Europa. Figure~\ref{fig:Europa_raw} shows a superposition of a basemap with a particular region highlighted. In addition one sees some fine structures on the icy surface which become visible when zooming in. By means of tailored classic and deep learning algorithms \citep{Haslebacher2024a,HaslebacherPSJ2025} one can extract data pairs $(\Xi,\Vi)$ as described before, see Figure~\ref{fig:Europa_processed}. The whole data set contains $5623$ observation pairs in $19$ different regions, but we only show $2184$ of these pairs. The region in the foreground (black dots, green lines) corresponds to the region with red boundary in Figure~\ref{fig:Europa_raw}. It contains $2656$ observations $(\Xi,\Vi)$, and we show a subset of $200$ pairs.

The task is to smooth these noisy data and to determine for various points $\x_o \in \S^2$ whether there is a preferred axis direction of ice cracks close to $\x_o$. The existence and strength of preferred directions allows conclusions about physical properties of the moon, e.g.\ the link of tidal forces to ice cracks \citep{Rhoden2013}. Translated into a mathematical task, we want to determine for such a point $\x_o$ a Bingham distribution \citep{Bingham_1974} on the unit circle $\S^2 \cap \x_o^\perp$ of the tangent space $\x_o^\perp$ of the unit sphere at $\x_o$.

\begin{figure}
\centering
\includegraphics[width=0.9\textwidth]{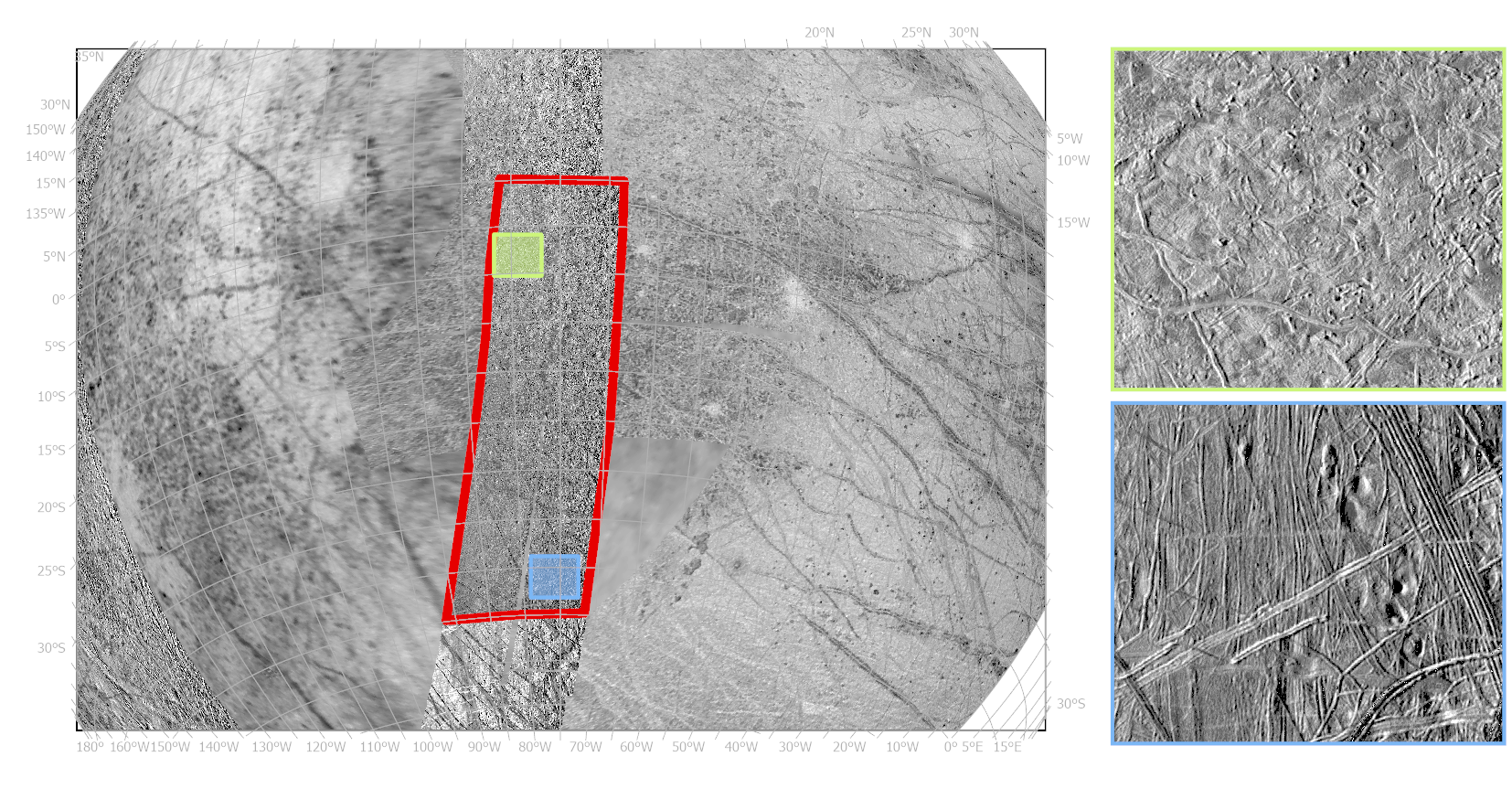}
\caption{An image of Europa's surface (left panel) and some fine structures.}
\label{fig:Europa_raw}
\end{figure}

\begin{figure}
(a) \hfill (b) \hfill \strut\\
\includegraphics[width=0.49\textwidth]{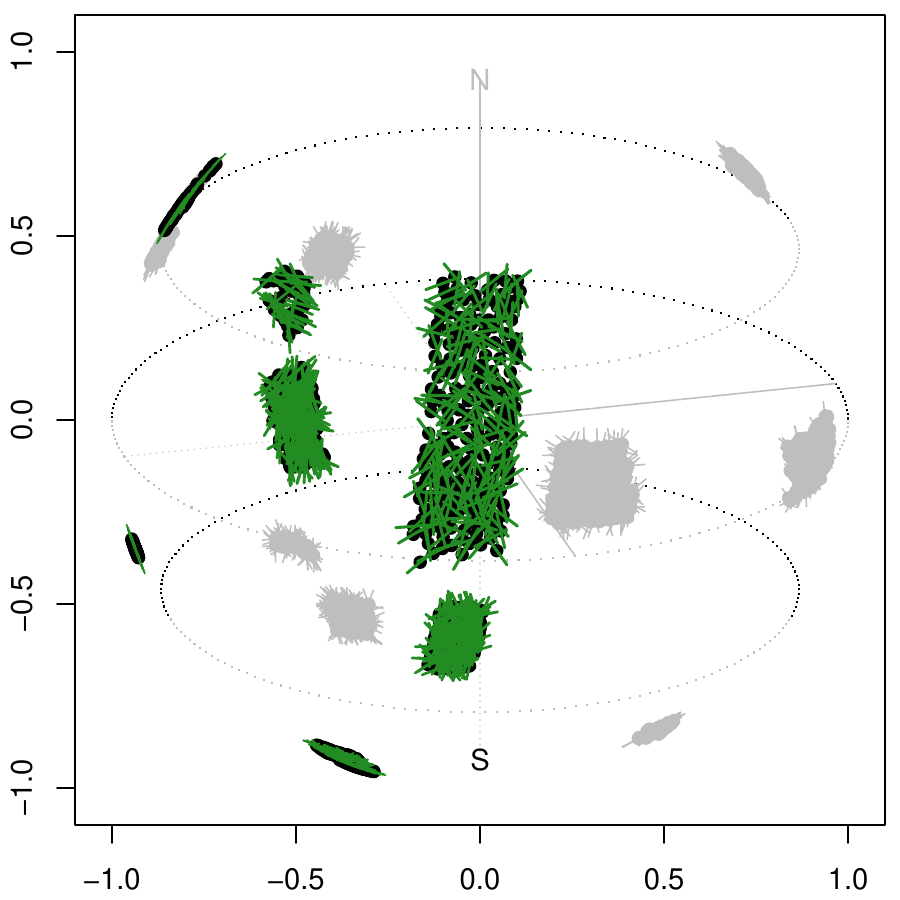}
\hfill
\includegraphics[width=0.49\textwidth]{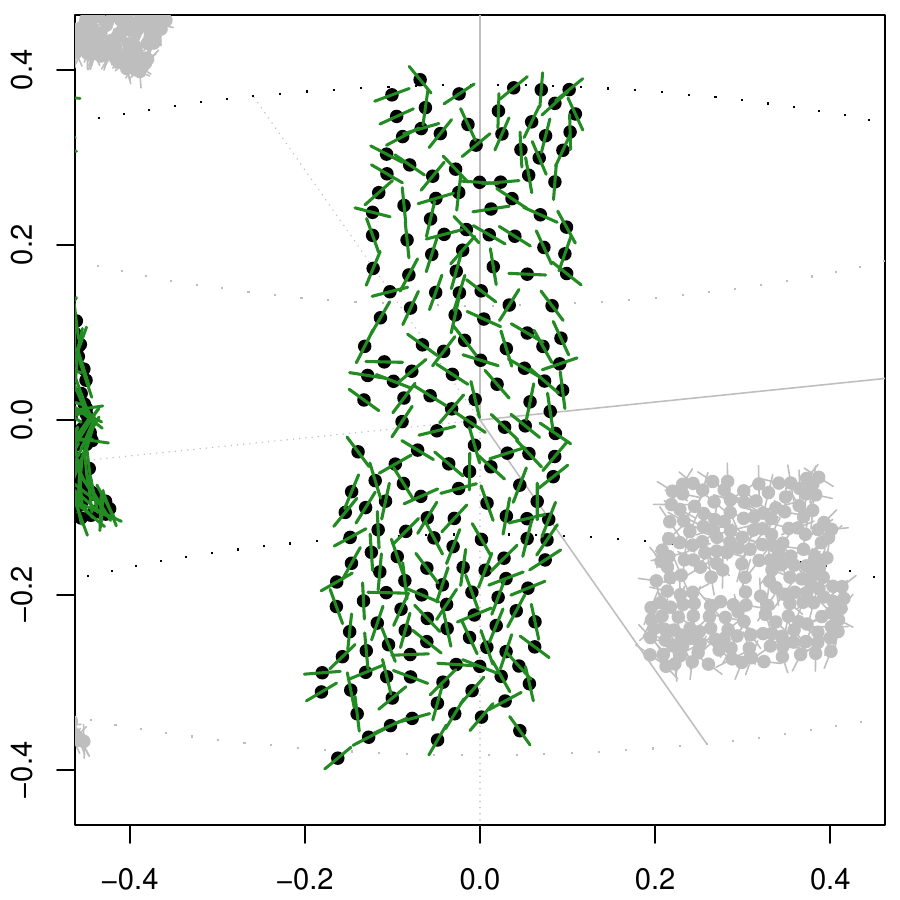}

\caption{(a) Processed image data $(\Xi,\Vi)$ for Europa's surface. The points $\Xi$ are shown as black dots, the axes $\Vi$ are indicated by green line segments connecting $\Xi \pm 0.05\cdot\Vi$. Observations on the backside are shown in gray. (b) The main cluster corresponds to the region highlighted in red in Figure~\ref{fig:Europa_raw}.}
\label{fig:Europa_processed}
\end{figure}

When thinking about smaller subregions of the unit sphere $\S^2$, one could approximate this by a subset of $\R^2$, and the task looks \textsl{similar} to the task of smoothing directional data, given by pairs $(\Xi,\Yi) \in \R^2 \times \S^1$, $1 \le i \le n$, where the conditional distribution of $\Yi$, given $\Xi$, is a von Mises--Fisher distribution with parameters depending smoothly on $\Xi$. More generally, one could think about observations $(\Xi,\Yi)$ with covariates $\Xi$ in an arbitrary covariate space $\XX$ and directional responses $\Yi \in \S^{d-1}$ for some $d \ge 2$.

Two general references about directional data, including regression methods, are the monograph of \cite{Mardia_Jupp_2000} and the review paper of \cite{Pewsey_Garcia-Portugues_2021}. These references describe and cite numerous approaches to regression with directional response. Some of these approaches involve transformations of the directional response into points in some Euclidean space, but the choice of this transformation is somewhat ad hoc. Let us restrict our attention to models in which $\Yi \in \S^{d-1}$ has a von Mises--Fisher distribution conditional on $\Xi$. Many authors characterize a von Mises--Fisher distribution via a direction parameter $\bs{\beta} \in \S^{d-1}$ and a concentration parameter $\kappa \ge 0$ which is often viewed as a nuisance parameter, whereas $\bs{\beta}$ is the parameter of interest. We propose to replace such a parameter pair $(\bs{\beta},\kappa)$ with the vector $\bs{z} = \kappa \bs{\beta} \in \R^d$. Since $\kappa = \|\bs{z}\|$ and $\bs{\beta} = \kappa^{-1} \bs{z}$, this description is equivalent, and dealing with parameters in a Euclidean space is more convenient, particularly in generalized linear models.

The remainder of this paper is structured as follows. Section~\ref{sec:Directional.axial} starts with a short recap of von Mises--Fisher (vMF) distributions and Bingham (Bh) distributions on the unit sphere $\S^{d-1}$ of $\R^d$. Then we describe a very useful connection between these two families for dimension $d = 2$. Section~\ref{sec:Regression} describes regression methods for data with directional response vectors. As a parametric approach, we propose generalized linear models (GLMs). For a thorough introduction and overview of this topic we refer to \cite{McCullagh_Nelder_1989}. Thereafter we describe a modification of this parametric approach to local parametric modelling in the spirit of \cite{Fan_etal_1998}. The latter method is illustrated with simulated data. In Section~\ref{sec:Europa} we explain how to modify the nonparametric methods of Section~\ref{sec:Regression} for smoothing axial data on a sphere, illustrating this method with the specific data mentioned in the beginning. A key tool are stereographic projections of the covariate vectors $\Xi \in \S^2$ onto $\R^2$ and suitable transformations of the accompanying axes $\Vi \in \S^2 \cap \Xi^\perp$ into points on the unit circle $\S^1$.

An appendix contains additional information and proofs. In particular, Appendix~\ref{app:Computational.aspects.vMF} provides details about the numerical computation of the mean vector and covariance matrix of a von Mises--Fisher distribution. This involves exact formulae as well as approximations for parameter vectors with large norm and is potentially of independent interest.

\section{Directional and axial Gaussian distributions}
\label{sec:Directional.axial}

For a given dimension $d \ge 2$, we consider the space $\R^d$ equipped with standard Euclidean norm $\|\cdot\|$. A directional distribution describes a random unit vector $\Y \in \S^{d-1}$, where $\S^{d-1} = \{\y \in \R^d : \|\y\| = 1\}$. An axial distribution describes a random unit vector $\V \in \S^{d-1}$ with symmetric distribution, where $\V$ represents the random line $\R\V$. In what follows, let $M$ denote the uniform distribution on $\S^{d-1}$.

\subsection{Von Mises--Fisher distributions}

The von Mises--Fisher distribution with parameter $\z \in \R^d$, denoted by $\vMF(\z)$, is given by the density
\[
	f_{\z}(\y) \ := \ e_{}^{\z^\top\y - \gamma(\z)}
\]
with respect to $M$, where $\gamma : \R^d \to \R$ is the cumulant-generating function of $M$,
\[
	\gamma(\z) \ := \ \log \int e_{}^{\z^\top\y} \, M(\d\y) .
\]
Some authors rather talk about the vMF distribution with concentration parameter $\|\z\|$ and center $\|\z\|^{-1} \z$, see \cite{Mardia_Jupp_2000}. But in view of the regression methods introduced later, we stick to the parametrization in terms of vectors in $\R^d$.

The family $(\vMF(\z))_{\z \in \R^d}$ is a natural exponential family \citep{Barndorff-Nielsen_2014}, and it is well-known that for a random vector $\Y \sim \vMF(\z)$, its mean vector and covariance matrix are given by the gradient and Hessian matrix of $\gamma$,
\begin{align*}
	\bmu(\z) := \Ex(\Y) \
	&= \ \nabla \gamma(\z) , \\
	\bSigma(\z) := \Cov(\Y) \
	&= \ D^2 \gamma(\z) .
\end{align*}
Other well-known facts are that the mapping $\bmu : \R^d \to \{\x \in \R^d : \|\x\| < 1\}$ is a diffeomorphism. Precisely,
\begin{equation}
\label{eq:mu.vMF}
	\bmu(t \vb) \ = \ \tilde{\gamma}_d'(t) \vb
	\quad\text{for} \ \vb \in \S^{d-1}, t \in \R ,
\end{equation}
where $\tilde{\gamma}_d : \R \to \R$ is an even, strictly convex and analytic function such that $\tilde{\gamma}_d(0) = 0$, and $\tilde{\gamma}_d' : \R \to (-1,1)$ is bijective, see Appendix~\ref{app:Computational.aspects.vMF}. By means of a Lagrange argument one can show that the density $f_{\z}$ maximizes differential Shannon entropy
\[
	- \int f(\y) \log f(\y) \, M(\d\y)
\]
among all probability densities $f$ with respect to $M$ satisfying
\[
	\int f(\y) \y \, M(\d\y) \ = \ \bmu(\z) .
\]
For that reason, some authors refer to the vMF distributions as Gaussian directional distributions, alluding to the fact that any Gaussian density maximizes entropy among all densities with the same first and second moments.

\subsection{Bingham distributions}

A random axis $\R\V$ in $\R^d$, given by a random unit vector $\V \in \S^{d-1}$ with symmetric distribution, can be identified with the symmetric matrix $\V\V^\top \in \R^{d\times d}$ describing the orthogonal projection of $\R^d$ onto $\R\V$. Maximizing differential Shannon entropy over all densities $f$ with respect to $M$ such that $\int f(\vb) \vb\vb^\top \, M(\d\vb)$ is a given symmetric matrix with trace $1$ leads to the following type of distribution: For a symmetric matrix $\bs{W} \in \R^{d\times d}$ with trace $0$, the Bingham distribution $\Bh(\bs{W})$ has density
\[
	f_{\bs{W}}^{\rm B}(\vb) \
	= \ e_{}^{\langle \vb\vb^\top, \bs{W}\rangle - \gamma^{\rm B}(\bs{W}} \
	= \ e_{}^{\vb^\top\bs{W}\vb - \gamma^{\rm B}(\bs{W})} ,
\]
where
\[
	\gamma^{\rm B}(\bs{W}) \ := \ \log \int e_{}^{\vb^\top\bs{W}\vb} \, M(\d\vb) .
\]
\cite{Bingham_1974} introduced these distributions for dimension $d = 3$, but the extension to arbitrary dimensions $d \ge 2$ is straightforward. Again, since the density of $\Bh(\bs{W})$ maximizes differential Shannon entropy among all probability densities $f$ such that $\int f(\vb) \vb\vb^\top \, M(\d\vb)$ is a given symmetric matrix with trace $1$, some authors refer to the Bh distributions as Gaussian axial distributions.

\subsection{The special case of $d = 2$}

In the special case of the unit circle $\S^1$, there is a useful connection between vMF and Bh distributions. Starting from an axis $\R\vb$ with $\vb \in \S^1$, we rescale and center the corresponding projection matrix $\vb\vb^\top$ similarly as \cite{Arnold_Jupp_2013} and write
\[
	2 \vb \vb^\top - \bs{I}_2 \
	= \ \begin{bmatrix} 2v_1^2 - 1 & 2 v_1^{} v_2^{} \\ 2 v_1^{} v_2^{} & 2 v_2^2 - 1 \end{bmatrix} \
	= \ \begin{bmatrix} v_1^2 - v_2^2 & 2 v_1^{} v_2^{} \\ 2 v_1^{} v_2^{} & v_2^2 - v_1^2 \end{bmatrix} \
	= \ \begin{bmatrix} y_1(\vb) & y_2(\vb) \\ y_2(\vb) & - y_1(\vb) \end{bmatrix} ,
\]
where
\begin{equation}
\label{eq:def.y(v)}
	\y(\vb) \ := \ [v_1^1 - v_2^2, 2 v_1^{} v_2^{}]^\top \ \in \ \mathbb{S}^1 .
\end{equation}
If we identify any vector in $\R^2$ with a complex number, then $\y(\vb) = \vb^2$. In particular, this shows immediately that $\|\y(\vb)\| = 1$. Hence there is a one-to-one correspondence between axes in $\R^2$ and the unit circle $\mathbb{S}^1$, and specifying the expectation $\int \vb\vb^\top \, Q(\d\vb)$ for a distribution $Q$ on $\S^1$ is equivalent to specifying the expectation $\int \y(\vb) \, Q(\d\vb)$. Note further that any symmetric matrix $\bs{W} \in \R^{2\times 2}$ with trace $0$ may be written as
\begin{equation}
\label{eq:W.and.z}
	\bs{W} \ = \ \begin{bmatrix} z_1 & z_2 \\ z_2 & -z_1 \end{bmatrix}
\end{equation}
for some $\z \in \R^d$, and elementary algebra reveals that
\[
	\vb^\top \bs{W} \vb \ = \ \z^\top \y(\vb) .
\]
Furthermore, if $\V \sim M$, then $\y(\V) \sim M$ too, and these findings imply the following very useful fact.

\begin{Proposition}
\label{prop:Bh.vMF}
With $\bs{W}$ and $\z$ as in \eqref{eq:W.and.z}, a random vector $\V$ follows $\Bh(\bs{W})$ if and only if $\y(\V)$ has distribution $\vMF(\z)$.
\end{Proposition}

To visualize a Bingham distribution $\Bh(\bs{W})$, let us resort to random angles. One may write
\[
	\bs{W} \ = \ \kappa
	\begin{bmatrix} \cos(2\beta) & \sin(2\beta) \\ \sin(2\beta) & -\cos(2\beta) \end{bmatrix}
\]
for some $\kappa \ge 0$ and some angle $\beta \in [0,\pi)$. If $\vb = [\cos(\theta), \sin(\theta)]^\top$ for some angle $\theta \in [0,2\pi)$, then $\y(\vb) = [\cos(2\theta), \sin(2\theta)]^\top$ and
\[
	\vb^\top \bs{W} \vb = \y(\vb)^\top \z \ = \ \kappa \cos(2 (\theta - \beta)) .
\]
Hence, $\V \sim \Bh(\bs{W})$ means that $\V = [\cos(\tilde{V}), \sin(\tilde{V})]^\top$ with a random angle $\tilde{V} \in [0,2\pi)$ having density
\begin{equation}
\label{eq:Bingham.angular}
	\tilde{f}_{\wb}(\theta) \ := \ e_{}^{\kappa \cos(2(\theta - \beta)) - \tilde{\gamma}_2(\kappa)}
\end{equation}
with respect to the uniform distribution on $[0,2\pi)$, where
\[
	\tilde{\gamma}_2(\kappa) \
	:= \ \log \Bigl( \frac{1}{2\pi}
		\int_0^{2\pi} e_{}^{\kappa\cos(2 (\theta - \beta))} \, \d\theta \Bigr) \
	= \ \log  \Bigl( \frac{1}{\pi}
		\int_0^{\pi} e_{}^{\kappa\cos(t)} \, \d t \Bigr) .
\]
This shows that in case of $\kappa > 0$, the preferred axis direction is $\pm [\cos(\beta), \sin(\beta)]^\top$, and $\kappa$ measures the strength of this preference. Thus we reparametrize the Bingham distribution $\Bh(\bs{W})$ with the parameter vector
\begin{equation}
\label{eq:def.Bh.w}
	\wb \ := \ \kappa [\cos(\beta), \sin(\beta)]^\top
\end{equation}
and write $\Bh(\wb)$ instead of $\Bh(\bs{W})$.

Figure~\ref{fig:Bingham} depicts the distribution $\Bh(\wb)$ for two different vectors $\wb = \kappa [\cos(\beta),\sin(\beta)]^\top$. One possible representation is a graph of the angular density $\tilde{f}$ in \eqref{eq:Bingham.angular}. Another possibility is to draw an ``axial histogram'', that is, the region surrounded by the curve $\theta \mapsto \tilde{f}(\theta)^{1/2} [\cos(\theta), \sin(\theta)]^\top$, leading to a symmetric but non-circular ``pie'' (unless $\kappa = 0$). The area of any ``slice'' of this pie with boundary angles $0 \le \theta_1 < \theta_2 \le 2\pi$ equals $2\pi$ times the probability that $\tilde{V} \in [\theta_1,\theta_2)$.

\begin{figure}
\includegraphics[width=0.49\textwidth]{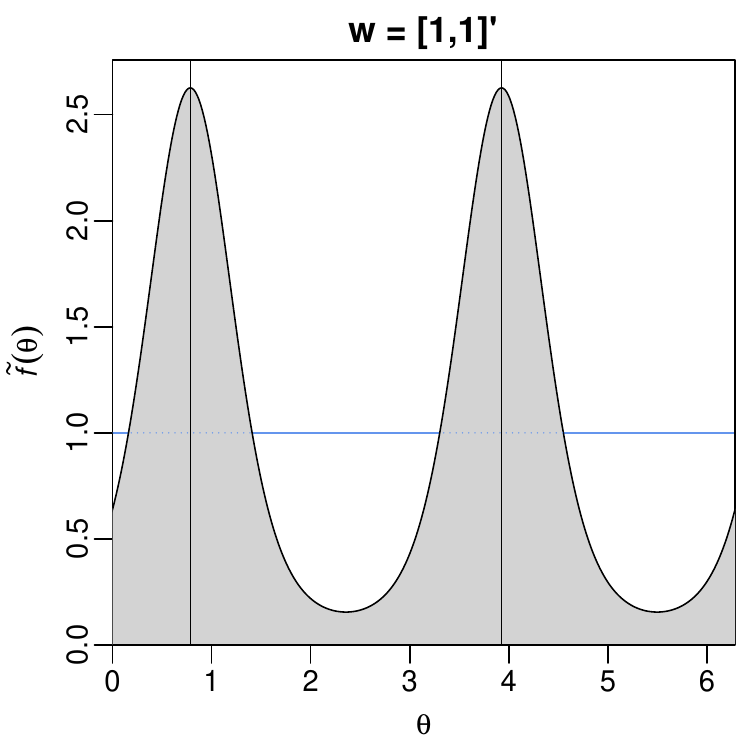}
\hfill
\includegraphics[width=0.49\textwidth]{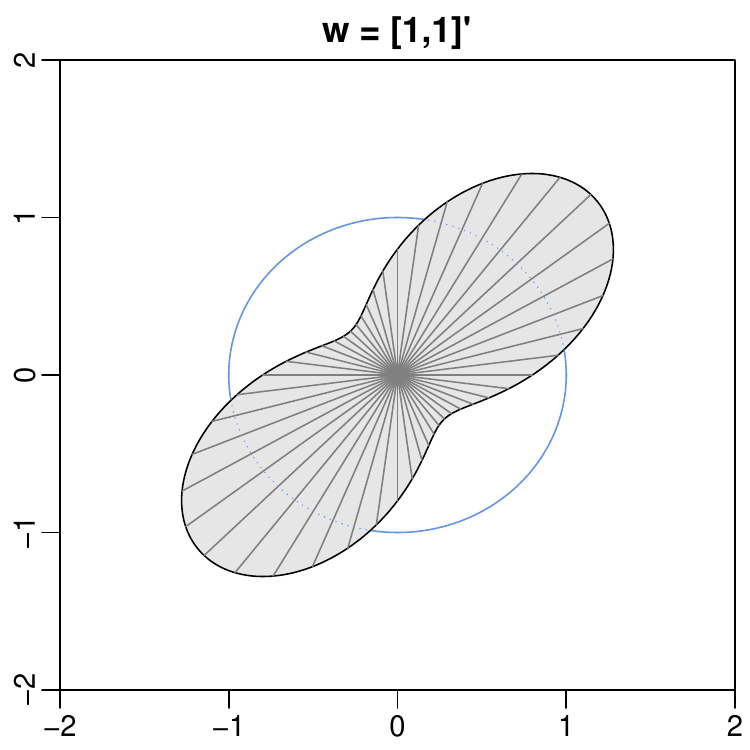}

\includegraphics[width=0.49\textwidth]{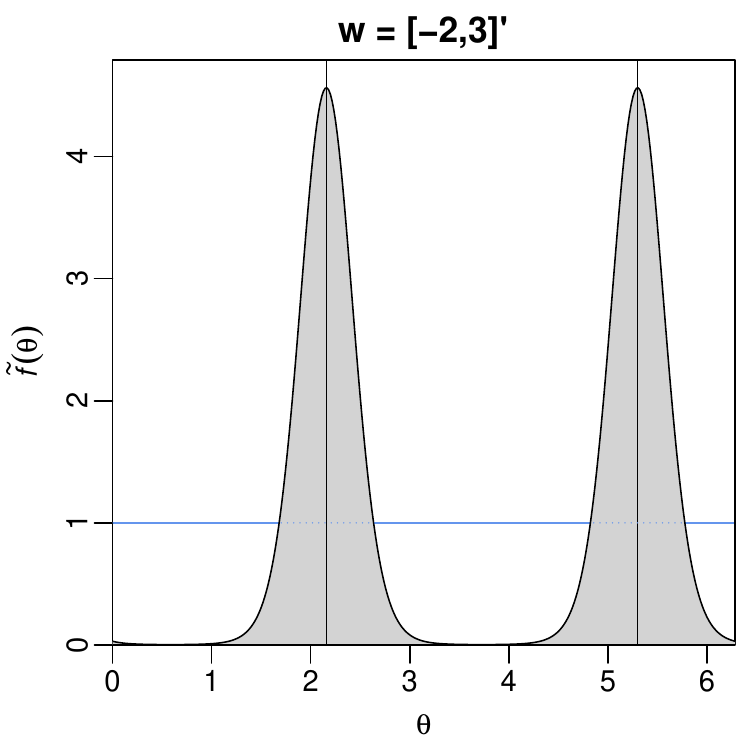}
\hfill
\includegraphics[width=0.49\textwidth]{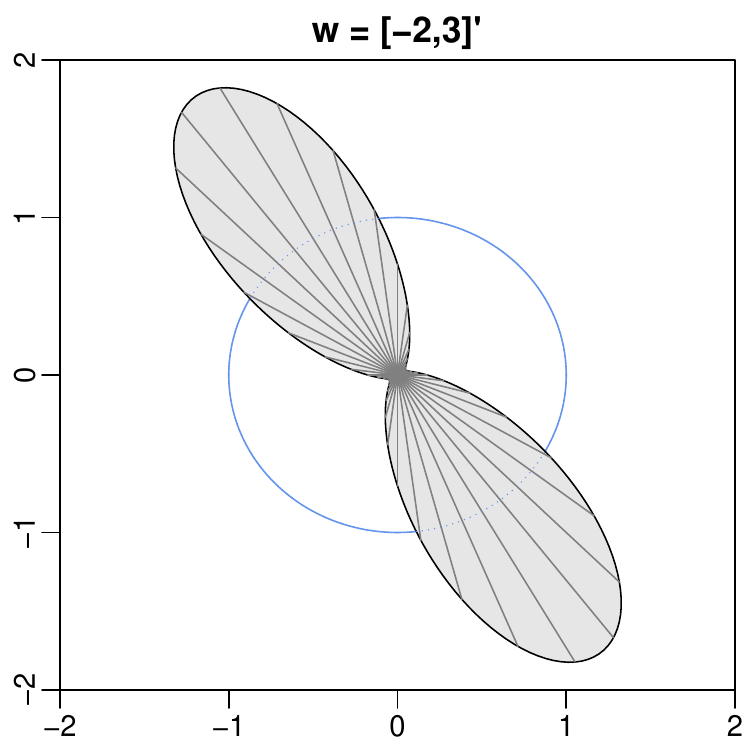}

\caption{Bingham distributions $\Bh(\wb)$ for $\wb = [1,1]^\top$ (top row, $\kappa \approx 1.414$, $\beta \approx 0.785$) and $\wb = [-2,3]^\top$ (bottom row, $\kappa \approx 3.606$, $\beta \approx 2.159$). The left column shows the angular densities $\tilde{f}$, the right column the ``axial histograms''.}
\label{fig:Bingham}
\end{figure}

The next proposition provides some additional formulae which will be used later and follow from Proposition~\ref{prop:Bh.vMF}, basic properties of von Mises--Fisher distributions and elementary calculations.

\begin{Proposition}
\label{prop:Moments.Bh.w}
Let $\V \sim \Bh(\wb)$ with $\wb$ as in \eqref{eq:def.Bh.w}. Then,
\[
	\Ex(\V\V^\top) \
	= \ 2^{-1} (1 - \tilde{\gamma}_2'(\kappa)) \bs{I}_2^{}
		+ \tilde{\gamma}_2'(\kappa) \bs{u} \bs{u}_{}^\top
\]
with $\bs{u} := [\cos(\beta),\sin(\beta)]^\top$, the preferred axis direction of $\Bh(\wb)$. Moreover, with $\|\cdot\|_F$ denoting Frobenius norm,
\[
	\Ex \bigl( \| \V\V^\top - \Ex(\V\V^\top) \|_F^2 \bigr) \
	\ = \ \frac{1 - \tilde{\gamma}_2'(\kappa)^2}{2} .
\]
\end{Proposition}

In view of this proposition, in the context of regression with Bingham-distributed response vectors, we visualize $\Bh(\wb)$ simply by the vector pair $\pm \tilde{\gamma}_2'(\kappa) \bs{u}$.

\section{Regression methods for directional data}
\label{sec:Regression}

Consider observation pairs $(\Xi,\Yi)$, $1 \le i \le n$, consisting of a covariate (vector) $\Xi$ in an arbitrary set $\XX$ and a response vector $\Yi \in \S^{d-1}$ such that
\[
	\LL(\Yi \,|\, \X_1, \ldots, \X_n) \ = \ \vMF(\bsf^*(\Xi))
\]
for some unknown regression function $\bsf^* : \XX \to \R^d$. The task is to estimate $\bsf^*$ from the given data.

\subsection{Generalized linear models (GLMs)}

A possible parametric approach is to assume that the unknown regression function $\bsf^*$ belongs to a given finite-dimensional space $\FF$ of functions $\bsf : \XX \to \R^d$. Conditioning on the covariates $\Xi$, $1 \le i \le n$, and assuming that the responses $\Yi$ are conditionally independent, the resulting negative log-likelihood function $\ell = \ell(\cdot \,|\, \mathrm{data}) : \FF \to \R$ is given by
\[
	\ell(\bsf) \ = \ \sum_{i=1}^n \bigl( \gamma(\bsf(\Xi)) - \Yi^\top \bsf(\Xi) \bigr) . 
\]
In our specific settings, any $\bsf \in \FF$ can be written as $\bsf = (f_k)_{k=1}^d$ with functions $f_k$ in a given finite-dimensional space $\FF^o$ of functions $f^o : \XX \to \R$. For this particular case, Appendix~\ref{app:Derivatives.weighted.NLL} provides technical details about the first and second derivatives of the function $\ell(\cdot)$ after a suitable parametrization of $\FF_o$. These formulae enable us to minimize $\ell(\bsf)$ over all $\bsf \in \FF$ via a Newton-Raphson procedure.

\subsection{Smoothing via local GLMs}

A possible extension of parametric GLMs are nonparametric analyses, where the unknown regression function $\bsf^*$ is only assumed to be ``smooth'' and estimated via local parametric models in the spirit of \cite{Fan_etal_1998}. The value of $\bsf^*$ at a particular point $\x^o \in \XX$ is estimated by $\hat{\bsf}(\x^o) = \hat{\bsf}_{\!\x_o}(\x^o)$, where
\[
	\hat{\bsf}_{\!\x_o}(\cdot) \ = \ \argmin_{\bsf \in \FF} \,
		\sum_{i=1}^n w_{\x_o}(\Xi) \bigl( \gamma(\bsf(\Xi)) - \Yi^\top \bsf(\Xi) \bigr)
\]
with a certain weight function $w_{\x_o}(\cdot) : \XX \to [0,\infty)$. In principle, one could even let the model $\FF$ depend on $\x_o$, but for our specific settings this is not done.

Specifically let $\XX = \R^q$. For a given number $N \in (0,n)$, one could look for the $N$-nearest neighbors of $\x_o$ among the observed vectors $\Xi$ and then fit a local polynomial model to the corresponding pairs $(\Xi,\Yi)$. That is, one could reorder the observations $(\Xi,\Yi)$ such that the distance $\|\Xi - \x_o\|$ is non-decreasing in $i$, and then set $w_{\x_o}(\x) = 1_{[\|\x - \x_o\| \le R_N]}$ with $R_N := \|\bs{X}_{\!N} - \x_o\|$. We propose a smooth version of this nearest-neighbor approach and set
\begin{equation}
\label{eq:local.weights}
	w_{\x_o}(\x) \ = \ \exp( - S_{\x^o,N} \|\x - \x^o\|^2) ,
\end{equation}
where $S_{\x^o,N} > 0$ is chosen such that $\sum_{i=1}^n \exp \bigl( - S_{\x^o,N} \|\Xi - \x^o\|^2) = N$.

If $\FF$ consists of all constant functions with values in $\R^d$, then the estimation task is easily solved by
\[
	\hat{\bsf}(\x^o) \ = \ \bmu_{}^{-1}(\bar{\Y}(\x_o))
\]
with the local mean
\[
	\bar{\Y}(\x_o) \ := \ \sum_{i=1}^n w_{\x_o}(\X_i) \Yi \Big/ \sum_{i=1}^n w_{\x_o}(\Xi)
\]
and the inverse $\bmu_{}^{-1} : \{\z \in \R^d : \|\z\| < 1\} \to \R^d$ of the mean function $\bmu(\cdot)$ for vMF distributions.

If $\FF$ consists of all linear functions $\bsf : \R^q \to \R^d$, then a suitable local basis of $\FF^o$ is given by the $r = q+1$ functions $f^o_0(\x) := 1$ and
\[
	f^o_j(\x) \ := \ x_j^{} - x_{o,j}^{}, \quad 1 \le j \le q ,
\]
with $\x_o = (x_{o,j})_{j=1}^q$.

If $\FF$ consists of all quadratic functions $\bsf : \R^q \to \R^d$, then the former $q+1$ local basis functions are complemented with the $q(q+1)/2$ functions
\[
	f^o_{jk}(\x) \ := \ (x_j^{} - x_{o,j})(x_k^{} - x_{o,k}) , \quad 1 \le j \le k \le q ,
\]
so $\FF^o$ has dimension $r = (q+1)(q+2)/2$.

\subsection{Numerical example and simulation study}

To illustrate these methods, we simulated $n = 4000$ independent observations $(\Xi,\Yi)$ $1 \le i \le n$, such that $\Xi$ is uniformly distributed on $[-1,1]^2$ while conditional on $\Xi$, $\Yi$ follows the vMF distribution with parameter $\bsf^*(\Xi)$, where
\[
	\bsf^*(\x) \ = \ \exp( - 2 \|\x\|^2) \begin{bmatrix} 1 \\ 3x_1 \end{bmatrix} .
\]
Figure~\ref{fig:GLM_data} depicts the regression function $\bmu(\bsf^*)$ as a vector field (left panel) and shows the raw data $(\Xi,\Yi)$ (right panel). For the true function, at each point $\x_o$ on the regular grid $\XX_o = \{k/10 : -10 \le k \le 10\}^2$ of $21^2 = 441$ points (gray bullets), a line segment connecting $\x_o$ with $\x_o + 0.18\cdot\bmu(\bsf^*(\x_o))$ is attached (black line). Similarly, at each location $\Xi$, a line segment connecting $\Xi$ and $\Xi + 0.05\cdot\Yi$ is attached.

\begin{figure}
\includegraphics[width=0.49\textwidth]{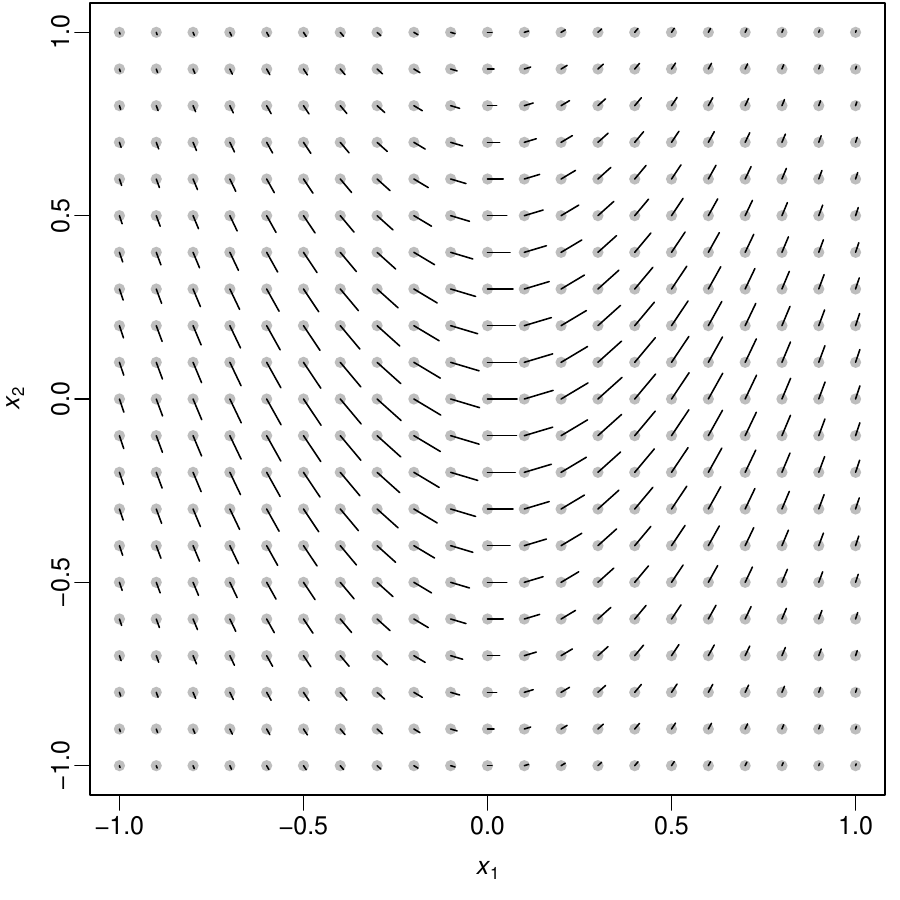}
\hfill
\includegraphics[width=0.49\textwidth]{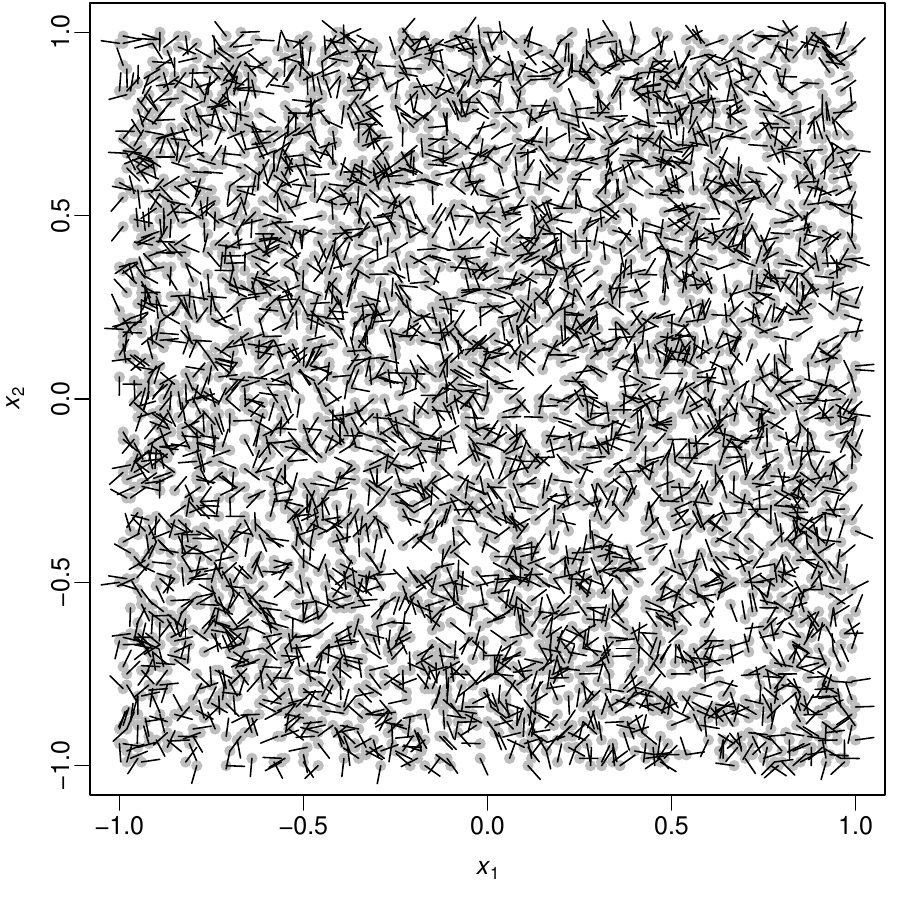}
\caption{True function $\bmu(\bsf^*)$ (left panel) and simulated raw data $(\Xi,\Yi)$ (right panel).}
\label{fig:GLM_data}
\end{figure}

Figure~\ref{fig:GLM_fhat} depicts estimated regression functions $\bmu(\hat{\bsf})$ for $N = 400$ together with the true regression function $\bmu(\bsf^*)$. Precisely, one sees the estimators based on local constant and local quadratic modelling. Looking at these plots carefully, one can see that the first estimator is more biased than the second one close to the boundary of $[-1,1]^2$. Moreover, the first estimator seems to underestimate the norm of $\bmu(\bsf^*)$ in the central region, whereas the second one is rather accurate there.

\begin{figure}
\includegraphics[width=0.49\textwidth]{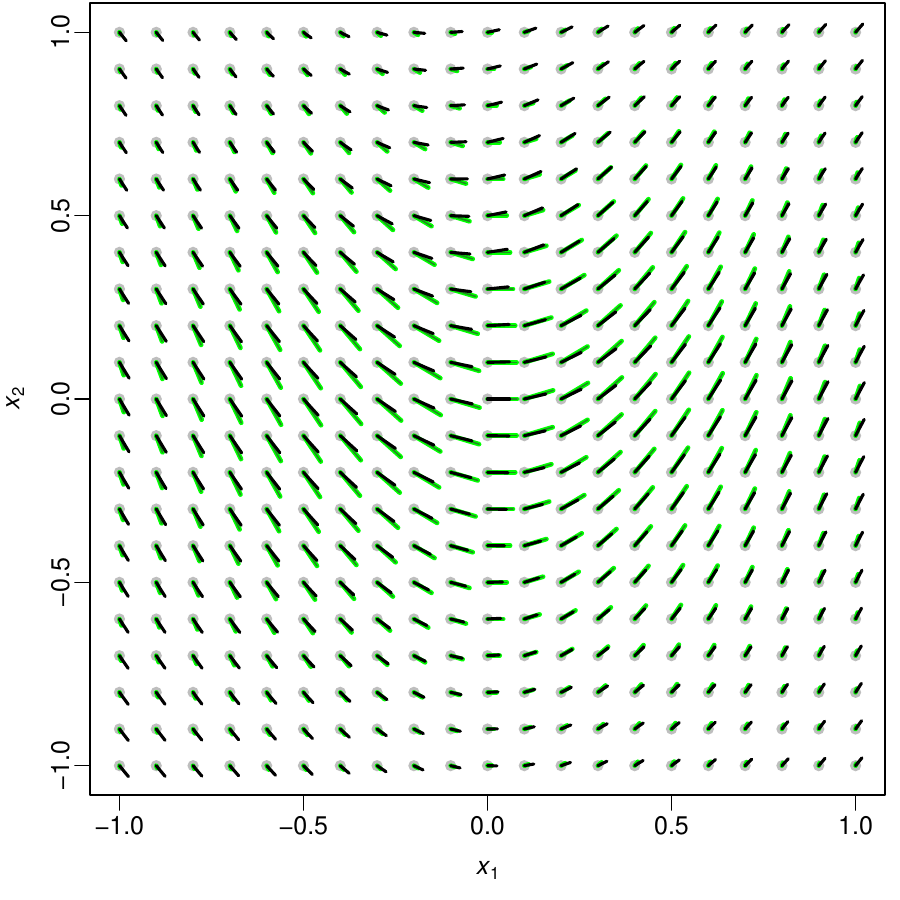}
\hfill
\includegraphics[width=0.49\textwidth]{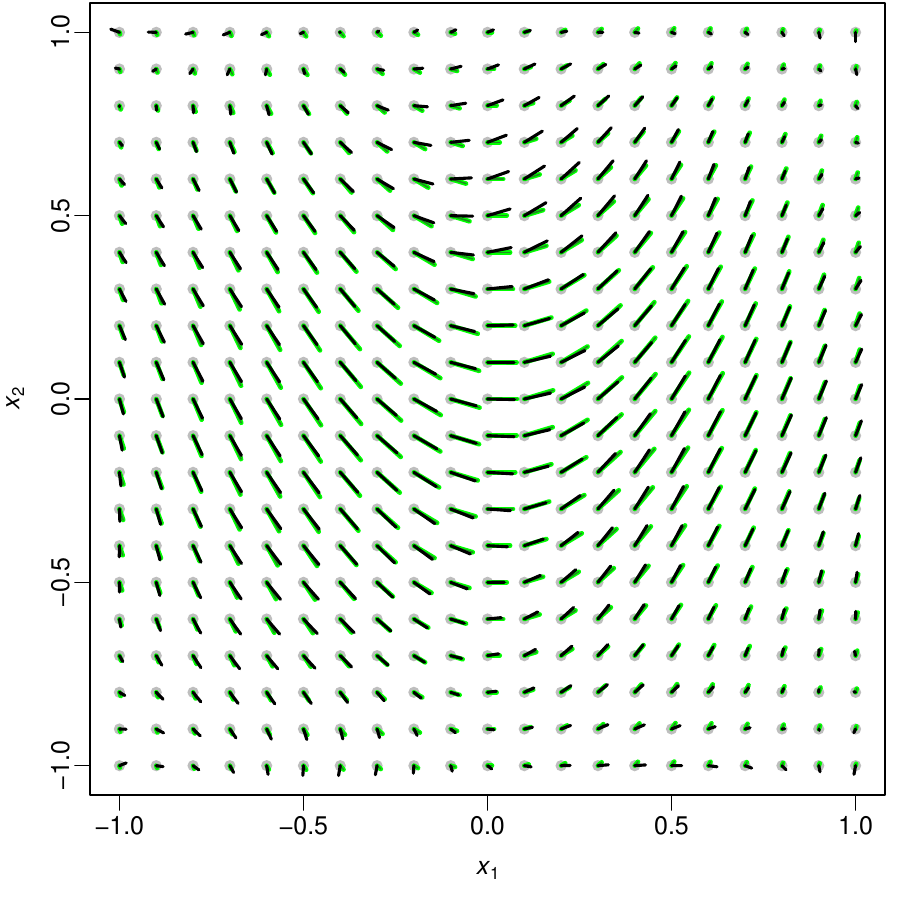}
\hfill

\caption{Estimated regression functions $\bmu(\hat{\bsf})$ based on local constant (left panel) and local quadratic (right panel) modelling in case of $N = 400$. In addition to the fit $\bmu(\hat{\bsf})$ (black) one sees $\bmu(\bsf^*)$ (green).}
\label{fig:GLM_fhat}
\end{figure}

Next we performed a little simulation study with $100$ simulations of such a data set. With these simulations we estimated for the three types of local polynomial estimates and different values of $N$ the following quantities:
\begin{align*}
	\mathrm{BIAS} \
	&:= \ \sqrt{ 441^{-1} \sum_{\x_o \in \XX_o}
		\bigl\| \Ex \bmu(\hat{\bsf}(\x_o)) - \bmu(\bsf^*(\x_o)) \bigr\|^2 } , \\
	\mathrm{SD} \
	&:= \ \sqrt{ 441^{-1} \sum_{\x_o \in \XX_o}
		\Ex \bigl[ \bigl\| \bmu(\hat{\bsf}(\x_o)) - \Ex \bmu(\hat{\bsf}(\x_o)) \bigr\|^2 \bigr] } , \\
	\mathrm{RMSE} \
	&:= \ \sqrt{ 441^{-1} \sum_{\x_o \in \XX_o}
		\Ex \bigl[ \bigl\| \bmu(\hat{\bsf}(\x_o)) - \bmu(\bsf^*(\x_o)) \bigr\|^2 \bigr] }
		\ = \ \sqrt{\mathrm{BIAS}^2 + \mathrm{SD}^2} .
\end{align*}
Table~\ref{tab:GLM} contains the results. These numbers show that we have similar effects as in local polynomial least squares regression. As the parameter $N$ increases, the bias increases while the variability decreases. For fixed $N$, the bias of local constant estimation is larger than the one for local linear estimation, and the latter is larger than the one for local quadratic estimation. The variability however, measured by $\mathrm{SD}$, increases with the model complexity.

\begin{table}
\[
	\begin{array}{|c||c|c|c||c|c|c||c|c|c|}
	\hline
		& \multicolumn{3}{|l||}{\text{constant}}
		& \multicolumn{3}{|l||}{\text{linear}}
		& \multicolumn{3}{|l|}{\text{quadratic}} \\
	N	& \mathrm{BIAS} & \mathrm{SD} & \mathrm{RMSE}
		& \mathrm{BIAS} & \mathrm{SD} & \mathrm{RMSE}
		& \mathrm{BIAS} & \mathrm{SD} & \mathrm{RMSE} \\
	\hline\hline
	100 & .032 & .069 & .076
		& .016 & .088 & .089
		& .013 & .130 & .130 \\
	\hline
	200 & .050 & .049 & .070
		& .029 & .065 & .071
		& .011 & .095 & .096 \\
	\hline
	300 & .066 & .040 & .077
		& .041 & .056 & .069
		& .014 & .079 & .080 \\
	\hline
	400 & .080 & .035 & .087
		& .054 & .050 & .074
		& .021 & .070 & .073 \\
	\hline
	500 & .091 & .033 & .097
		& .066 & .047 & .081
		& .025 & .065 & .070 \\
	\hline
	600 & .101 & .030 & .105
		& .075 & .045 & .088
		& .032 & .061 & .068 \\
	\hline
	700 & .109 & .028 & .113
		& .086 & .042 & .096
		& .035 & .058 & .067 \\
	\hline
	800 & .117 & .026 & .120
		& .095 & .040 & .103
		& .039 & .055 & .068 \\
	\hline
	\end{array}
\]
\caption{Estimated error measures for different local polynomial estimators and values $N$.}
\label{tab:GLM}
\end{table}

In Appendix~\ref{app:Further.details.numerical.example}, we show also graphical displays of the estimated pointwise bias
\[
	\Ex \bmu(\hat{\bsf}(\x_o)) - \bmu(\bsf^*(\x_o)) .
\]

\section{Smoothing Axial Data on a Sphere}
\label{sec:Europa}

We return to a data set consisting of points $(\Xi,\Vi)$, $1 \le i \le n$, where $\Xi \in \S^2$ and $\Vi \in \S^2 \cap \Xi^\perp$ representing the axis $\R \Vi$ in the tangent plane of the sphere at $\Xi$. By means of these observations, we want to fit for any point $\x_o \in \S^3$ a Bingham distribution $\Bh(\x_o,\hat{\bsf}(\x_o))$, where $\hat{\bsf}(\x_o) \in \x_o^\perp$. For any vector $\wb \in \x_o^\perp$, the Bingham distribution $\Bh(\x_o,\wb)$ is defined as follows: Let $\bs{e}_1, \bs{e}_2$ be an orthonormal basis of $\x_o^\perp$ such that the orthogonal matrix $[\x_o,\bs{e}_1,\bs{e}_2]$ has determinant $1$. Writing $\wb = \kappa \cos(\beta) \bs{e}_1 + \kappa \sin(\beta) \bs{e}_2$ for some $\kappa \ge 0$ and $\beta \in [0,\pi)$, the distribution $\Bh(\x_o,\wb)$ describes the distribution of the random vector
\[
	\V \ = \ \cos(\tilde{V}) \bs{e}_1 + \sin(\tilde{V}) \bs{e}_2 \ \in \ \x_o^\perp ,
\]
where $\tilde{V} \in [0,2\pi)$ is a random variable following the density \eqref{eq:Bingham.angular}. To relate the observations $(\Xi,\Vi)$ to the point $\x_o \in \S^2$ and its tangent plane $\x_o^\perp$, we need a suitable transformation which is described in the next two subsections.

\subsection{A stereographic projection}

We start with the particular reference point $\x_o = [1,0,0]^\top$. Any point $\x \in \R^3$ with $x_1 > -1$ is mapped to a point $P(\x) \in \R^2$ as follows: One moves $\x$ along the straight line connecting the reference point's antipode $[-1,0,0]^\top$ and $\x$ such that it hits the hyperplane $\{\z \in \R^3 : z_1 = 1\}$. That is, we need $\nu(\x) \in \R$ such that $(1 - \nu(\x)) [-1,0,0]^\top + \nu(\x) \x = [1, P(\x)^\top]^\top$. This leads to
\[
	P(\x) \ := \ \nu(\x) \begin{bmatrix} x_2 \\ x_3 \end{bmatrix}
	\quad\text{with}\quad
	\nu(\x) \ := \ \frac{2}{1 + x_1} .
\]
If restricted to $\XX := \S^2 \setminus \{[-1,0,0]^\top\}$, the mapping $P : \XX \to \R^2$ is a diffeomorphism, and its inverse mapping $P^{-1} : \R^2 \to \XX$ is given by
\[
	P^{-1}(\z) \ := \ \begin{bmatrix} 2 \omega(\z) - 1 \\ \omega(\z) z_1 \\ \omega(\z) z_2 \end{bmatrix}
	\quad\text{with}\quad
	\omega(\z) \ := \ \frac{4}{4 + \|\z\|^2} .
\]
For $\x \in \XX$, any vector $\vb \in \x^\perp$ can be viewed as the derivative of a smooth curve in $\S^2$ passing through $\x$, so it is natural to consider $DP(\x) \vb$ with the Jacobian matrix $DP(\x) \in \R^{2\times 3}$. Elementary calculations show that
\[
	DP(\x) \ = \ \nu(\x) \bs{A}(\x)
	\quad\text{with}\quad
	\bs{A}(\x) \ := \ \begin{bmatrix}
		- x_2/(1 + x_1) & 1 & 0 \\
		- x_3/(1 + x_1) & 0 & 1
	\end{bmatrix} .
\]
It is well-known that $P$ is a conformal mapping in the sense that $DP(\x) : \x^\perp \to \R^2$ preserves angles. Precisely, one can show that for arbitrary $\vb \in \x^\perp$,
\[
	\|\bs{A}(\x) \vb\|^2 \ = \ \|\vb\|^2 .
\]
Thus we relate the pair $(\x,\vb)$ to the reference point $[1,0,0]^\top$ and its tangent plane by mapping it to the pair
\[
	\bigl( P(\x), \bs{A}(\x)\vb \bigr) \in \R^2 \times \R^2 .
\]

Figure~\ref{fig:Stereographic} illustrates this mapping. The upper left panel shows an artificial set of points $(\Xi,\Vi)$, $1 \le i \le n$, where $\Xi \in \S^2$ and $\Vi \in \S^2 \cap \Xi^\perp$. The axes correspond to the second and third components of all vectors. Note that the data points $\Xi$ (black dots) are situated on a finite collection of circles on $\S^2$, and each axis $\Vi$ (indicated by a green line connecting $\Xi \pm 0.1 \cdot \Vi$) is perpendicular to the circle containing $\Xi$. The other panels show the projected points $\bigl( P(\Xi), \bs{A}(\Xi)\Vi \bigr)$ in squares of different size centered around $\bs{0}$. One sees clearly the conformal nature of this projection and the well-known fact that circles are mapped onto circles.

\begin{figure}
\includegraphics[width=0.48\textwidth]{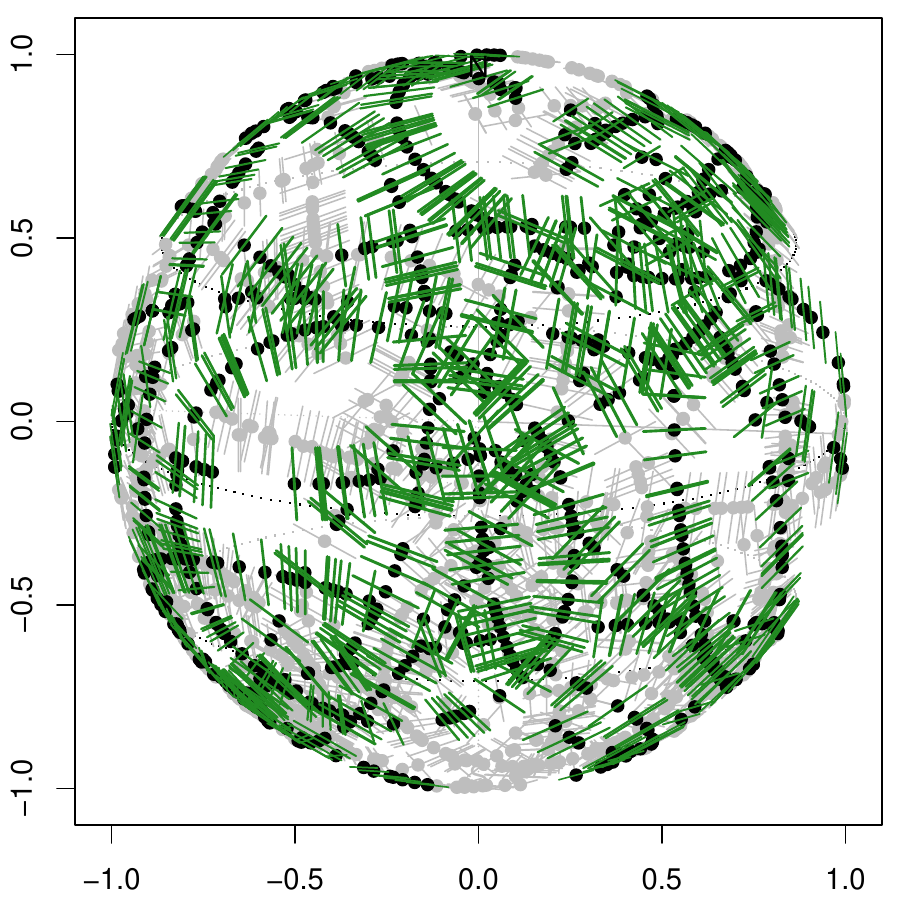}
\hfill
\includegraphics[width=0.48\textwidth]{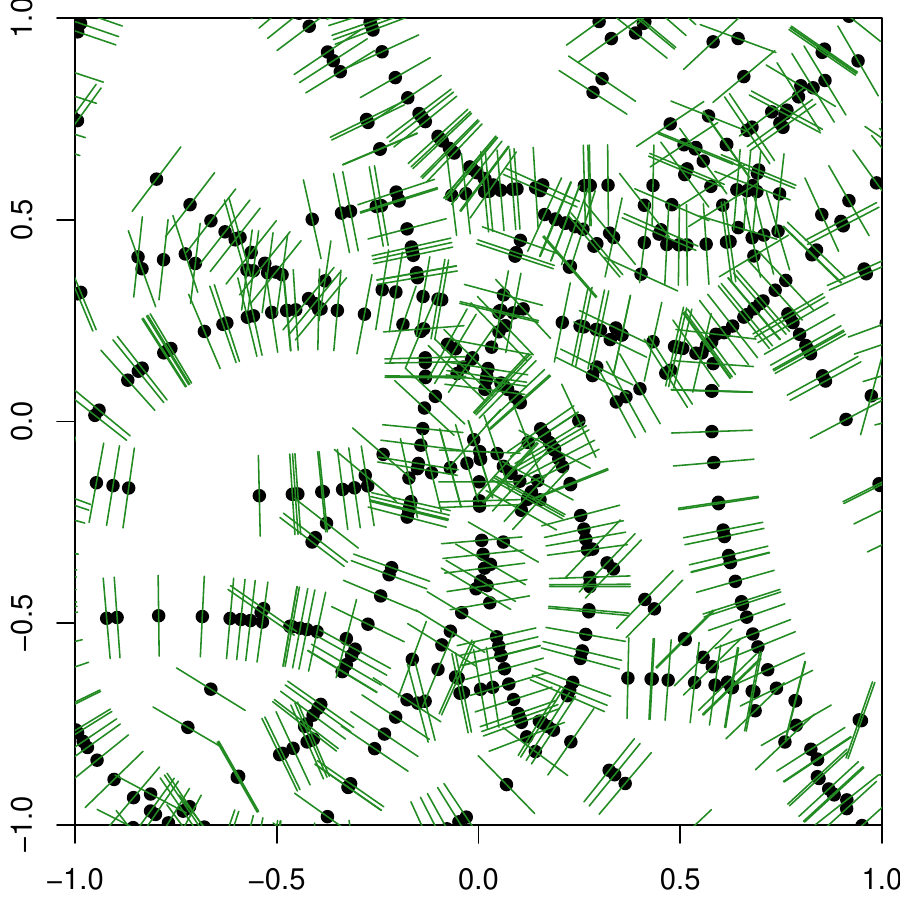}

\includegraphics[width=0.48\textwidth]{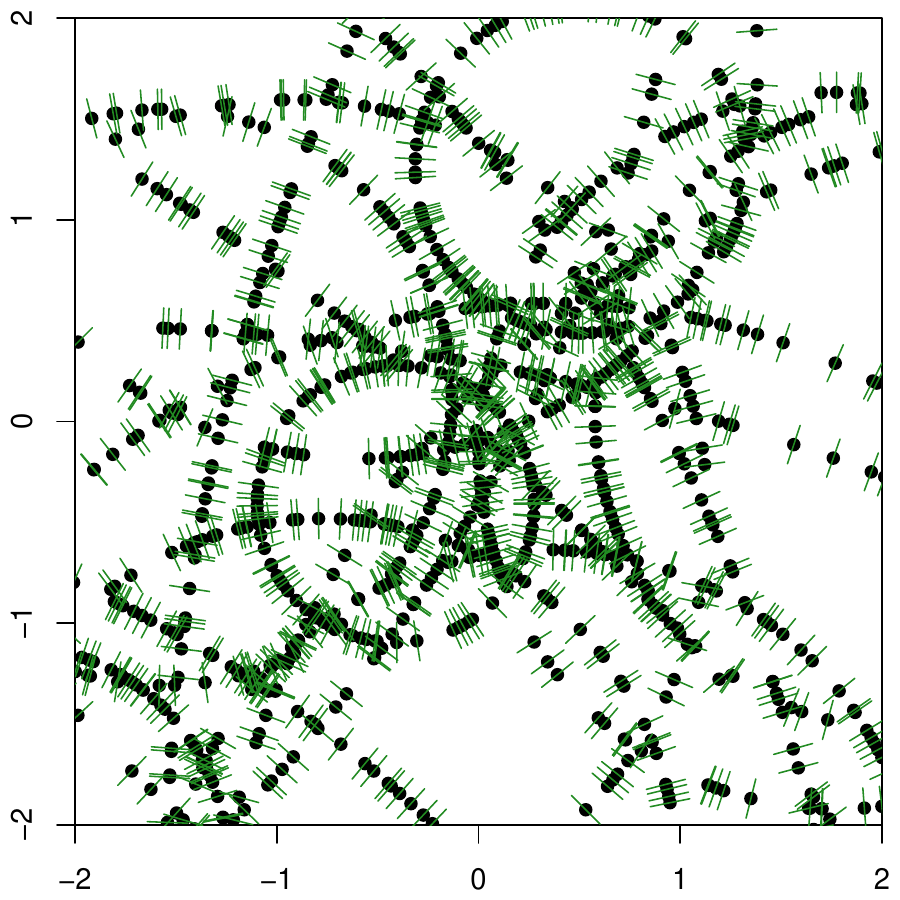}
\hfill
\includegraphics[width=0.48\textwidth]{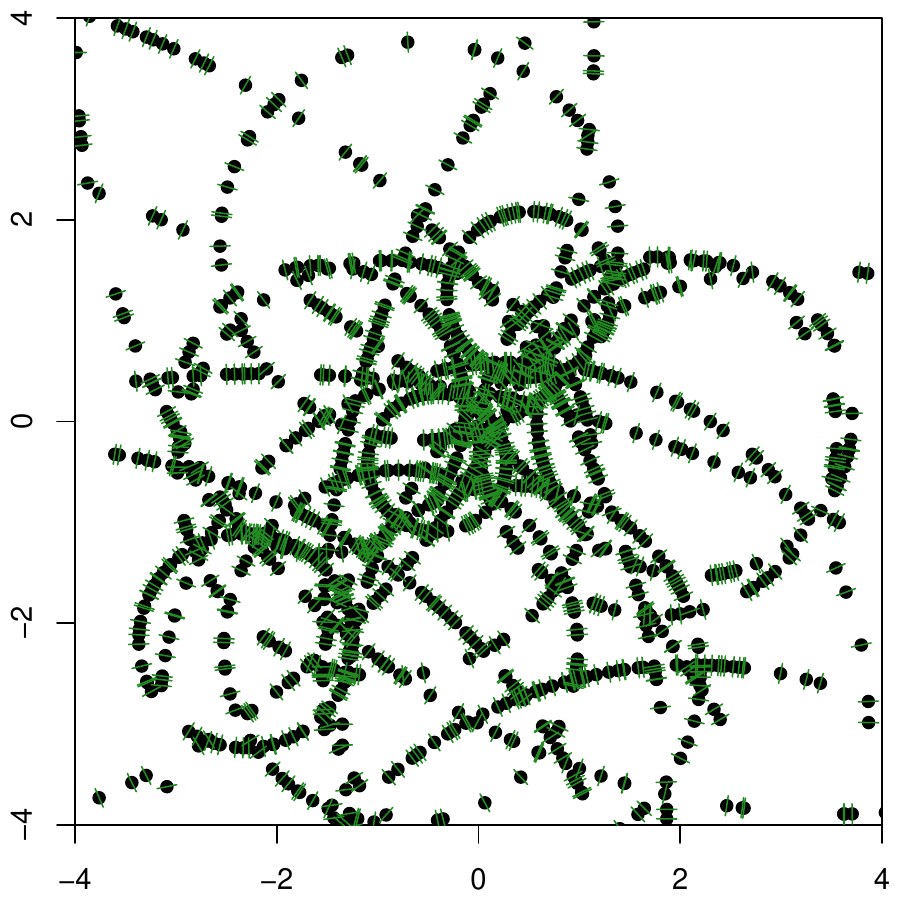}

\caption{Stereographic projections of an artificial data set. Top left: full sphere. Others: stereographic projection at different scales, showing that it is conformal.}
\label{fig:Stereographic}
\end{figure}

\subsection{Smoothing the data for an arbitrary reference point}

For an arbitrary reference point $\x_o \in \S^2$, we choose an orthogonal matrix $\bs{B}_o$ with determinant $1$ such that $\bs{B}_o\x_o = [1,0,0]^\top$. Then all data points $(\Xi,\Vi)$ with $\Xi \ne - \x_o$ are transformed into the pairs
\[
	\bigl( P(\bs{B}_o\Xi), \bs{A}(\bs{B}_o\Xi) \bs{B}_o \Vi \bigr) \in \ \R^2 \times \S^1 .
\]
Since $\bs{A}(\bs{B}_o\Xi) \bs{B}_o \Vi$ represents an axis in $\R^2$, we replace it by means of the mapping $\y(\cdot)$ in \eqref{eq:def.y(v)} with the direction $\y(\bs{A}(\bs{B}_o\Xi) \bs{B}_o \Vi) \in \S^1$. Then we apply the local polynomial estimators described in Section~\ref{sec:Regression} to the data pairs $\bigl( P(\bs{B}_o\Xi), \y(\bs{A}(\bs{B}_o\Xi) \bs{B}_o \Vi) \bigr) \in \R^2 \times \S^1$ and $\bs{0}$ in place of $(\Xi,\Yi)$ and $\bs{x}_o$. This yields an estimated vMF parameter $\hat{\z} \in \R^2$. Writing $\hat{\z} = \hat{\kappa} [\cos(2\hat{\beta}), \sin(2\hat{\beta})]^\top$ with $\hat{\kappa} \ge 0$ and $\hat{\beta} \in [0,\pi)$, we replace $\hat{\z}$ with the axis parameter $\hat{\wb} := \hat{\kappa} [\cos(\hat{\beta}), \sin(\hat{\beta})]^\top$ and transform it back to the Bingham parameter
\[
	\hat{\bsf}(\x_o) \ := \ \bs{B}_o^\top \begin{bmatrix} 0 \\ \hat{\wb} \end{bmatrix} \ \in \ \x_o^\perp .
\]
This leads to the estimate $\Bh(\x_o,\hat{\bsf}(\x_o))$ of $\Bh(\x_o,\bsf^*(\x_o))$. Note that any choice of $\bs{B}_o$ would yield the same estimate $\Bh(\x_o,\hat{\bsf}(\x_o))$.

\subsection{Some numerical results for Europa}

We analyzed the data from Europa and estimated Bingham distributions $\Bh(\x_o,\hat{\bsf}(\x_o))$ at many different locations $\x_o$. Precisely, within each of the $19$ regions, we chose an evenly spread subset of size up to $200$ of all observed locations $\Xi$ there. Now we show some results for the particular region in Figure~\ref{fig:Europa_processed}. We analyzed the data with $N = 50, 100, 150, 200, 300, 400$. As explained later, there is some evidence for overfitting (undersmooting) when $N = 50$ and underfitting (oversmoothing) when $N = 400$. For $N = 100, 150, 200, 300$, the pictures look similar and lead to the same conclusions. Figure~\ref{fig:Europa_fhat_200} shows the estimators resulting from local linear and local quadratic models with $N = 200$. For the $200$ selected points $\x_o$, the fitted Bingham distributions $\Bh(\x_o,\hat{\bsf}(\x_o))$ are represented by a blue line segment connecting the points
\[
	\x_o \pm 0.1 \cdot \tilde{\gamma}_2'(\hat{\kappa}(\x_o)) \hat{\bs{u}}(\x_o) ,
\]
where $\hat{\kappa}(\x_o)$ and $\hat{\bs{u}}(\x_o)$ are the norm and direction of $\hat{\bsf}(\x_o)$, respectively. Recall that for each estimate $\hat{\bsf}(\x_o)$, a new stereographic projection with reference point $\x_o$ was used.

In the upper part of the chosen region, the direction of the ice cracks seems to be rather chaotic (i.e.\ uniform), whereas in the middle and lower parts, there are preferred axis directions. These findings supplement observations by the physicists \citep{HaslebacherPSJ2025} who analysed the region divided into geological chaos terrain (upper part) and ridged plains (middle and lower part) by \cite{Leonard2024}. It is an interesting finding that in the chaos terrain, no ``order'' is found in the sense of a preferred crack direction. To which extent tidal forces enforce cracking is still under debate.

\begin{figure}
\includegraphics[width=0.49\textwidth]{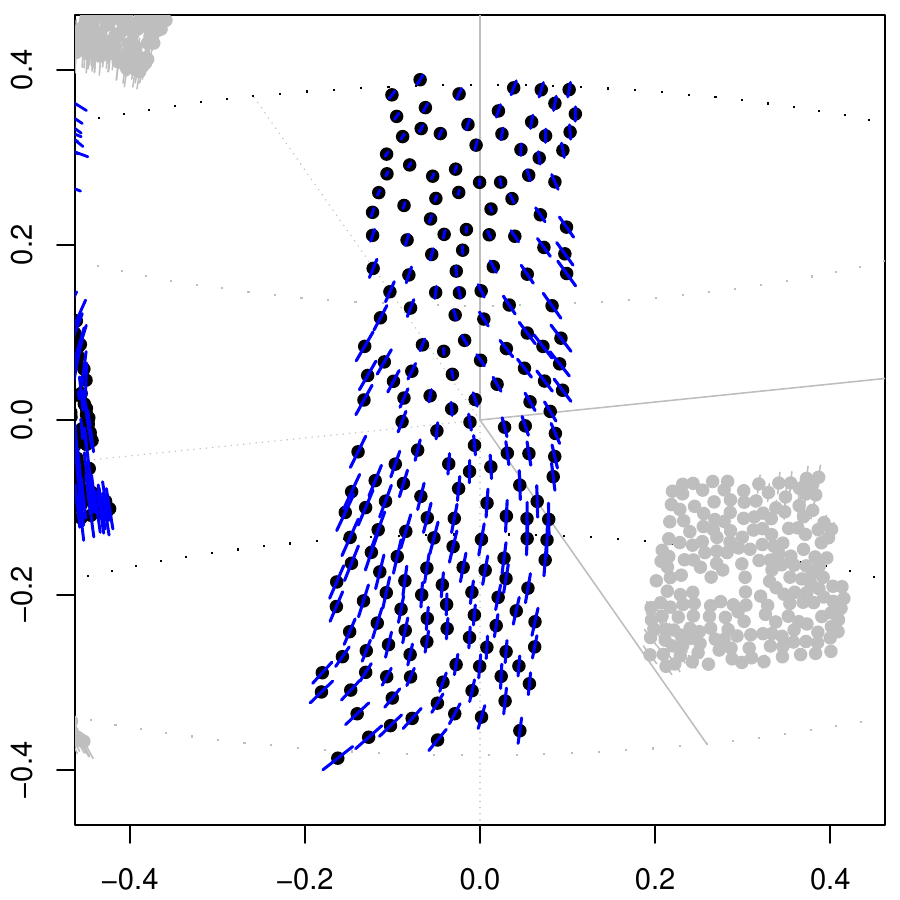}
\hfill
\includegraphics[width=0.49\textwidth]{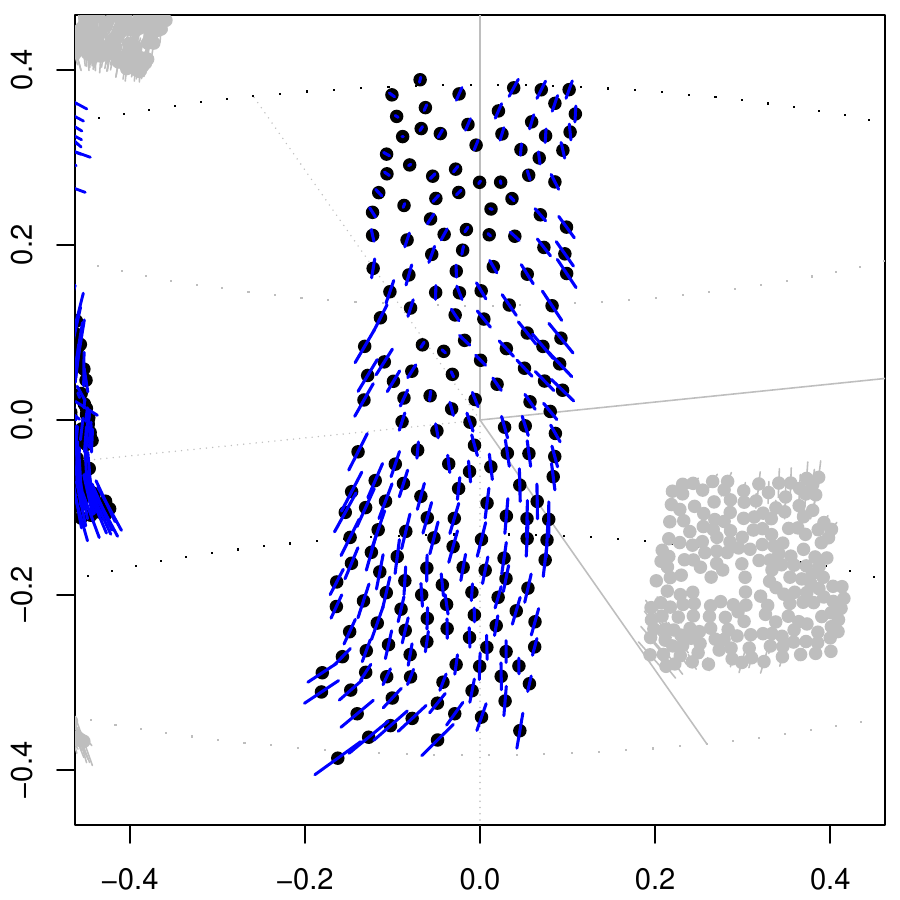}

\caption{Smoothed data from Europa with $N = 200$: Fitted regression function at $200$ locations via local linear (left panel) and local quadratic (right panel) models.}
\label{fig:Europa_fhat_200}
\end{figure}

An interesting open problem for this type of data analysis is regression diagnostics. As a first attempt, we propose to check the plausibility of the estimators by comparing two diagnostic quantities which can be seen as a surrogate for R-squared in least squares regression: Let $\XX_o = \{\Xi : i \in J_o\}$ be the set of $m_o = 200$ points $\x_o$ in the given region at which we compute $\hat{\bsf}(\x_o)$. Then we set
\[
	R_{\rm model}^2 \ := \ \frac{1}{n_o} \sum_{\x_o \in \XX_o} \tilde{\gamma}_2'(\hat{\kappa}(\x_o))^2
	\ \in \ [0,1] .
\]
This number measures how well the response vectors $\Vi$ can be predicted by the estimated regression model, assuming the latter to be true. Indeed, note that by Proposition~\ref{prop:Moments.Bh.w},
\[
	\tilde{\gamma}_2'(\kappa)^2 \
	= \ 1 - 2 \Ex \bigl( \| \V\V^\top - \Psi(\x_o, \kappa\bs{u}) \|_F^2 \bigr)
\]
for $\V \sim \Bh(\x_o,\kappa \bs{u})$, $\kappa \ge 0$, $\bs{u} \in \mathbb{S}^2 \cap \x_o^\perp$, where $\Psi(\x_o,\kappa \bs{u}) = \Ex(\V\V^\top)$. Thus an alternative measure of determination is given by
\[
	R_{\rm residual}^2 \
	:= \ \frac{1}{n_o} \sum_{i \in J_o}
		\bigl( 1 - 2 \bigl\| \V_{\!\!i}^{}\Vi^\top - \Psi(\Xi,\hat{\bsf}(\Xi)) \bigr\|_F^2 \bigr) .
\]
The ratio $R_{\rm residual}^2 / R_{\rm model}^2$ should be close to one, and higher or smaller values indicate over- or underfitting, respectively.

Not surprisingly, for $N = 50$, the values of $R_{\rm model}^2$ and $R_{\rm residual}^2$ are the largest, but the ratio $R_{\rm residual}^2/R_{\rm model}^2$ is $1.266$ for local constant models, which indicates overfitting. By way of contrast, for $N = 400$, the same ratio is $0.952$, indicating slight underfitting. For $N = 200$, the diagnostics are $(R_{\rm residual}^2, R_{\rm model}^2) = (0.117, 0.104)$ for local linear models (ratio = 1.122) and $(R_{\rm residual}^2, R_{\rm model}^2) = (0.137, 0.119)$ for local quadratic models (ratio = 1.152).

\subsection{Final comments}

The starting point for the present manuscript was the data from Europa. But as mentioned before, analogous questions arise with observations from Ganymede \citep{Rossi2020}. In the latter case, the axes refer to grooves (i.e.\ tectonic deformations) rather than cracks. We are currently experimenting with data from Ganymede and planning to compare our results with previous ones. There are also potential applications for planets with non-icy surfaces, e.g.\ faults in the solid surface of Mercury \citep{Watters2001} or Venus \citep{Sabbeth2023}.

In the derivation of the regression and smoothing methods, we used the standard setting of stochastically independent responses, given the covariates. For the application to Europa or Ganymede, this assumption is certainly not satisfied, for instance, because very long ice cracks lead to several observations $(\Xi,\Vi)$ with different locations $\Xi$ but similar axis directions $\Vi$. Thus, the methods presented here are merely exploratory and serve to find and visualize patterns in the data, without statistical conclusions such as standard errors, p-values or confidence bounds.

The simulations and data analyses were carried out with the programming language $\mathsf{R}$ \citep{R2023}. The code we used is availabe upon request.

\appendix

\section*{Appendix}

\section{Computational aspects of vMF distributions}
\label{app:Computational.aspects.vMF}

\paragraph{The function $\gamma$.}
If $\Yzero$ is a random vector with uniform distribution $M = \vMF(\bs{0})$ on $\S^{d-1}$, then for any unit vector $\vb \in \S^{d-1}$, the random variable $U_0 := \Yzero^\top \vb \in (-1,1)$ has density $h_d$,
\begin{equation}
\label{eq:density.hd}
	h_d(u) \ := \ C_d (1 - u^2)^{a_d - 1} , \quad u \in (-1,1) ,
\end{equation}
where $a_d := (d-1)/2$ and $C_d := 1/B(1/2,a_d)$ with the beta function $B(\cdot,\cdot)$. This implies that $\gamma(\z) = \log \int e_{}^{\z^\top\y} \, M(\d\y)$ can be written as
\[
	\gamma(\z) \ = \ \tilde{\gamma}_d(\|\z\|) ,
\]
where
\[
	\tilde{\gamma}_d(t) \ := \ \log G_d(t) , \quad
	G_d(t) \ := \ \int_{-1}^1 e^{tu} h_d(u) \, \d u
\]
for $t \in \R$. The function $G_d$ can also be written as $G_d(t) = t^{1-d/2} J_{d/2-1}(t)$, where $J_\beta$ stands for the modified Bessel function of the first kind, see \cite{Mardia_Jupp_2000}. In the special case $d = 3$, $h_3 \equiv 1/2$ and $G_3(t) = \sinh(t)/t$.

\begin{proof}[\bf Proof of \eqref{eq:density.hd}]
Let $\bs{Z} \in \R^d$ be a standard Gaussian random vector. Then, $\Y_0^\top \vb$ has the same distribution as $Z_1/\|\bs{Z}\|$, and $Z_1^2/\|\bs{Z}\|^2$ follows $\mathrm{Beta}(1/2,a_d)$. Denoting the density of the latter distribution with $b_d$, it follows from the symmetry of the distribution of $Z_1/\|\bs{Z}\|$ that for $u \in (0,1)$,
\[
	h_d(u) \
	= \ - \frac{\d}{\d u} \Pr(Z_1/\|\bs{Z}\| > u) \
	= \ - \frac{\d}{\d u} 2^{-1} \Pr(Z_1^2/\|\bs{Z}\|^2 > u^2) \
	= \ u b_d(u^2) ,
\]
whereas for $u \in (-1,0)$,
\[
	h_d(u) \
	= \ \frac{\d}{\d u} \Pr(Z_1/\|\bs{Z}\| < u) \
	= \ \frac{\d}{\d u} 2^{-1} \Pr(Z_1^2/\|\bs{Z}\|^2 > u^2) \
	= \ |u| b_d(u^2) .
\]
In both cases the result equals $C_d (1 - u^2)^{a_d-1}$, because $b_d(x) = C_d x^{-1/2} (1 - x)^{a_d-1}$ for $x \in (0,1)$.
\end{proof}

With our regression applications in mind, we decided to calculate $G_d$ directly, rather than using the detour via Bessel functions. A Taylor series for $G_d$ is given by
\begin{equation}
\label{eq:Taylor.Gd}
	G_d(t) \ = \ \sum_{k=0}^\infty c_{d,k} t^{2k}
	\quad\text{with}\quad
	c_{d,k} \
	:= \ \frac{2^{-2k} \, \Gamma(d/2)}{k! \, \Gamma(k + d/2)} .
\end{equation}

\begin{proof}[\bf Proof of \eqref{eq:Taylor.Gd}]
Since $h_d$ is even, we may rewrite $G_d(t) = \int_{-1}^1 e^{tu} h_d(u) \, \d u$ as
\begin{align*}
	G_d(t) \
	&= \ C_d \int_{-1}^1 \cosh(tu) (1 - u^2)^{a_d-1} \, \d u \\
	&= \ 2 C_d \int_0^1 \cosh(tu) (1 - u^2)^{a_d-1} \, \d u \\
	&= \ 2 C_d \sum_{k=0}^\infty \frac{t^{2k}}{(2k)!}
		\int_0^1 u^{2k} (1 - u^2)^{a_d-1} \, \d u \\
	&= \ C_d \sum_{k=0}^\infty \frac{t^{2k}}{(2k)!}
		\int_0^1 s^{k-1/2} (1 - s)^{a_d-1} \, \d s
		& (s = u^2, 2 \, \d u = s^{-1/2} \, \d s) \\
	&= \ \sum_{k=0}^\infty c_{d,k} t^{2k} ,
\end{align*}
where
\begin{align*}
	c_{d,k} \
	:= \ &\frac{1}{(2k)!} \, \frac{B(k+1/2,a_d)}{B(1/2,a_d)} \\
	= \ &\frac{1}{(2k)!} \, \frac{\Gamma(k+1/2)}{\Gamma(1/2)}
		\, \frac{\Gamma(d/2)}{\Gamma(k + d/2)} \\
	= \ &\frac{2^{-(2k-1)} \, \Gamma(2k)}{(2k)! \, \Gamma(k)}
		\, \frac{\Gamma(d/2)}{\Gamma(k + d/2)} \\
	= \ &\frac{2^{-(2k-1)} (2k-1)!}{(2k)! (k-1)!}
		\, \frac{\Gamma(d/2)}{\Gamma(k + d/2)} \\
	= \ &\frac{2^{-2k} \, \Gamma(d/2)}{k! \, \Gamma(k + d/2)} ,
\end{align*}
where the second step is a consequence of Legendre's well-known duplication formula for the gamma function.
\end{proof}

\paragraph{Mean and Covariance of $\vMF(\z)$.}
Note first that symmetry considerations reveal that $\Yzero \sim M = \vMF(\bs{0})$ satisfies
\[
	\Ex(\Yzero) \ = \ \bs{0}
	\quad\text{and}\quad
	\Ex(\Yzero^{}\Yzero^\top) \ = \ d^{-1} \bs{I}_d .
\]
An arbitrary $\z \in \R^d$ may be written as $\z = t \vb$ with $t = \|\z\|$ and some $\vb \in \S^{d-1}$. Then, $\Yz \sim \vMF(\bs{z})$ may be represented as
\[
	\Yz \ = \ U_t \vb + \sqrt{1 - U_t^2} \, \bs{S}_{\vb} ,
\]
where $U_t \in (-1,1)$ and $\bs{S}_{\vb} \in \S^{d-1}$ are stochastically independent, $U_t$ has density $e^{tu - \tilde{\gamma}_d(t)} h_d(u)$ at $u \in (-1,1)$, and $\bs{S}_{\vb}$ is uniformly distributed on the unit sphere $\S^{d-1} \cap \vb^\perp$ of the $(d-1)$-dimensional space $\vb^\perp$. In particular, one can deduce from the properties of $\Yzero$ that $\Ex(\bs{S}_{\vb}) = \bs{0}$ and $\Ex(\bs{S}_{\vb}^{} \bs{S}_{\vb}^\top) = (d-1)^{-1} (\bs{I}_d - \vb \vb^\top)$. Together with independence of $U_t$ and $\bs{S}_{\vb}$ we obtain the following formulae for $\bmu(\z)$ and $\bSigma(\z)$:
\begin{align}
\label{eq:mu(z)}
	\bmu(\z) \
	&= \ \Ex(U_t) \vb , \\
\label{eq:Sigma(z)}
	\bSigma(\z) \
	&= \ \Var(U_t) \, \vb\vb^\top + \frac{1 - \Ex(U_t^2)}{d-1} (\bs{I}_d - \vb\vb^\top) .
\end{align}
Since the distributions of $U_t$, $t \in \R$, form an exponential family with natural parametrization, one can write $\Ex(U_t) = \tilde{\gamma}_d'(t)$ and $\Var(U_t) = \tilde{\gamma}_d''(t)$, which leads to \eqref{eq:mu.vMF}. The fact that $\tilde{\gamma}_d' : [0,\infty) \to [0,1)$ is bijective is well-known and follows also implicitly from Corollary~\ref{cor:moments.Ut}. Note also that
\[
	\Ex(U_t^\ell) \
	= \ \int_{-1}^1 u^\ell e^{tu} h_d(u) \, \d u \Big/ \int_{-1}^1 e^{tu} h_d(u) \, \d u \
	= \ G_d^{(\ell)}(t) / G_d(t)
\]
with $G_d^{(\ell)}$ denoting the $\ell$-th derivative of $G_d$. Thus,
\begin{align*}
	\bmu(\z) \
	&= \ \frac{G_d'}{G_d}(t) \, \vb , \\
	\bSigma(\z) \
	&= \ \Bigl( \frac{G_d''}{G_d}(t) - \Bigl(\frac{G_d'}{G_d}(t) \Bigr)^2 \Bigr) \, \vb\vb^\top
		+ \frac{1}{d-1} \Bigl( 1 - \frac{G_d''}{G_d}(t) \Bigr) (\bs{I}_d - \vb\vb^\top) .
\end{align*}

\paragraph{Numerical computation of $G_d^{(\ell)}(t)$ for moderate values of $t$.}
Note first that
\[
	\frac{c_{d,k+1} t^{2(k+1)}}{c_{d,k} t^{2k}} \
	= \ \frac{(t/2)^2}{(k+1)(k+d/2)}
		\quad\text{for} \ k \ge 0 .
\]
Consequently,
\[
	\sum_{k=k_o+1}^\infty c_{d,k} t^{2k} \ \le \ \epsilon
	\quad\text{as soon as}\quad
	c_{d,k_o} t^{2k_o} \ \le \ \epsilon , \ \
	\frac{(t/2)^2}{(k_o+1)(k_o+d/2)} \ \le \ 1/2 .
\]
This allows the computation of $G_d(t)$ with arbitrary prescribed precision.

Any derivative $G_d^{(\ell)}$ can be expressed in terms of $G_{d+2m}$, $m=1,\ldots,\ell$. Indeed,
\begin{equation}
\label{eq:Gd1}
	G_d'(t) \
	= \ \frac{t}{d} \, G_{d+2}(t) .
\end{equation}
From this formula, one can proceed inductively. In particular,
\[
	G_d''(t) \
	= \ \frac{1}{d} \, G_{d+2}(t) + \frac{t^2}{d(d+2)} \, G_{d+4}(t) .
\]

\begin{proof}[\bf Proof of \eqref{eq:Gd1}]
Starting from the Taylor series of $G_d$,
\begin{align*}
	G_d'(t) \
	&= \ \sum_{k=1}^\infty 2k c_{d,k} t^{2k-1} \\
	&= \ \sum_{k=1}^\infty \frac{2k 2^{-2k} \,\Gamma(d/2)}{k! \, \Gamma(k + d/2)} \, t^{2k-1} \\
	&= \ \sum_{k=1}^\infty \frac{2^{-2k+1} \,\Gamma(d/2)}{(k-1)! \, \Gamma(k + d/2)} \, t^{2k-1} \\
	&= \ \frac{t \, \Gamma(d/2)}{2\,\Gamma(d/2 + 1)}
		\sum_{k=1}^\infty \frac{2^{-2(k-1)}\,\Gamma((d+2)/2)}{(k-1)!\,\Gamma((k-1) + (d+2)/2)}
			\, t^{2(k-1)} \\
	&= \ \frac{t}{d} \, G_{d+2}(t) .
\end{align*}\\[-5ex]
\end{proof}

\paragraph{Computation of $G_d(t)$, $\Ex(U_t)$ and $\Ex(U_t^2)$ for large values of $t$.}
For large values of $t$, computing $G_d(t)$, $G_{d+2}(t)$ and $G_{d+4}(t)$ via their series expansion becomes problematic. Here we resort to the following approximation formulae.

\begin{Lemma}
\label{lem:gammaz.muz.Sigmaz}
Let $\z = t \vb$ with $t > 0$ and $\vb \in \S^{d-1}$. Then as $t \to \infty$ and uniformly in $\vb$,
\begin{align*}
	\gamma(\z) \
	&= \ \log \Bigl( \frac{2^{a_d-1} \Gamma(d/2)}{\Gamma(1/2)} \Bigr)
		+ t - a_d \log(t) - \frac{a_d(a_d-1)}{2t} \Bigl( 1 + \frac{1}{2t} + O(t^{-2}) \Bigr) , \\
	\bmu(\z) \
	&= \ \Bigl( 1 - \frac{a_d}{t} \Bigl( 1 - \frac{a_d-1}{2t} + O(t^{-2}) \Bigr) \Bigr) \, \vb , \\
	\bSigma(\z) \
	&= \ \frac{a_d}{t^2} \Bigl( 1 - \frac{a_d-1}{t} + O(t^{-2}) \Bigr) \, \vb\vb^\top
		+ \frac{1}{t} \Bigl( 1 - \frac{a_d}{t} + O(t^{-2}) \Bigr) (\bs{I}_d - \vb\vb^\top) .
\end{align*}
\end{Lemma}

The formulae in this lemma, without the terms $O(t^{-2})$, provide excellent approximations as soon as $t = \|\z\|$ is larger than, say, $100 d$.

For the proof of Lemma~\ref{lem:gammaz.muz.Sigmaz}, we use a general result about expansions of moments for distributions on $(-1,1)$.

\begin{Lemma}
\label{lem:expansions}
Let $h : (-1,1) \to \R$ be integrable such that for some $k \in \mathbb{N}_0$ and constants $a > 0$ and $b_0,\ldots,b_k \in \R$,
\[
	h(1 - r) \ = \ \sum_{j=0}^k b_j r^{a+j-1} + O(r^{a+k})
\]
as $r \downarrow 0$. Then, for $\ell \in \mathbb{N}_0$,
\[
	\int_{-1}^1 (1 - u)^\ell e^{tu} h(u) \, \d u \
	= \ e_{}^t t_{}^{-(\ell+a)}
		\Bigl( \sum_{j=0}^k b_j \Gamma(\ell+a+j) t^{-j} + O(t^{-(k+1)}) \Bigr)
\]
as $t \to \infty$.
\end{Lemma}

\begin{proof}[\bf Proof of Lemma~\ref{lem:expansions}]
We may write
\[
	\int_{-1}^1 (1 - u)^\ell e^{tu} h(u) \, \d u \
	= \ e^t \int_{-1}^1 (1 - u)^\ell e^{-t(1-u)} h(u) \, \d u \
	= \ e^t \int_0^2 r^\ell e^{-tr} h(1 - r) \, \d r .
\]
For any fixed $\delta \in (0,2)$,
\[
	\int_0^2 r^\ell e^{-tr} h(1 - r) \, \d r \
	= \ \int_0^\delta r^\ell e^{-tr} h(1 - r) \, \d r + R_1(\delta,t)
\]
with
\[
	|R_1(\delta,t)| \
	\le \ 2^\ell \int_{-1}^1 |h(u)| \, \d u \, e^{-t\delta} \
	= \ O(e^{-t\delta}) .
\]
Moreover,
\[
	\int_0^\delta r^\ell e^{-tr} h(1 - r) \, \d r \
	= \ \sum_{j=0}^k b_j \int_0^\delta r^{\ell+a+j-1} e^{-tr} \, \d r
		+ R_2(\delta,t) ,
\]
where
\[
	|R_2(\delta,t)| \
	\le \ D(\delta) \int_0^\infty r^{\ell+a+k} e^{-tr} \, \d r \
	= \ D(\delta) \Gamma(\ell+a+k+1) t^{-(\ell+a+k+1)} ,
\]
and
\[
	D(\delta) \ := \ \sup_{r \in (0,\delta)} r^{-(a+k)} \Bigl| h(1 - r) - \sum_{j=0}^k b_j r^{a+j-1} \Bigr|
\]
is finite for sufficiently small $\delta > 0$.

Finally,
\begin{align*}
	\int_0^\delta r^{\ell+a+j-1} e^{-tr} \, \d r \
	&= \ t^{-(\ell+a+j)} \int_0^{t\delta} y^{\ell+a+j-1} e^{-y} \, \d y \\
	&= \ \Gamma(\ell+a+j) t^{-(\ell+a+j)} + O(t^{-(\ell+a+j)} e^{-\delta t/2}) ,
\end{align*}
because for any $m > 0$ and a random variable $Y$ with distribution $\mathrm{Gamma}(m,1)$, it follows from Markov's inequality that
\[
	\int_{t\delta}^\infty y^{m-1} e^{-y} \, \d y \
	= \ \Gamma(m) \Pr(G_m \ge t\delta) \
	\le \ \Gamma(m) \Ex(e^{G_m/2}) e^{-t\delta/2} \
	= \ \Gamma(m) 2^m e^{-t\delta/2} .
\]\\[-5ex]
\end{proof}

Lemma~\ref{lem:expansions} will be applied to the particular density $h_d$, noting that
\begin{align}
\nonumber
	h_d(1 - r) \
	&= \ C_d(2r - r^2)^{a_d-1} \\
\nonumber
	&= \ 2^{a_d-1} C_d r^{a_d-1} (1 - r/2)^{a_d-1} \\
\label{eq:expansion.hd}
	&= \ \frac{2^{a_d-1} \Gamma(d/2)}{\Gamma(a_d) \Gamma(1/2)} r^{a_d-1}
		\Bigl( 1 - \frac{a_d - 1}{2} r + \frac{[a_d - 1]_2}{8} r^2 + O(r^3) \Bigr)
\end{align}
as $r \downarrow 0$. Here and throughout this section, we use the notation $[s]_j := \prod_{i=0}^{j-1} (s-i)$ for real numbers $s$ and integers $j \ge 1$. Applying Lemma~\ref{lem:expansions} to $h_d$ with the latter expansion leads to expansions for $G_d(t)$ and $\log G_d(t)$.

\begin{Corollary}
\label{cor:Gd.logGd}
As $t \to \infty$,
\[
	G_d(t) \ = \ \frac{2^{a_d-1} \Gamma(d/2)}{\Gamma(1/2)} \, e_{}^t t_{}^{-a_d}
		\Bigl( 1 - \frac{[a_d]_2}{2t} + \frac{[a_d + 1]_4}{8t^2} + O(t^{-3}) \Bigr)		
\]
and
\[
	\log G_d(t) \ = \ \log \Bigl( \frac{2^{a_d-1} \Gamma(d/2)}{\Gamma(1/2)} \Bigr)
		+ t - a_d \log(t) - \frac{[a_d]_2}{2t} \Bigl( 1 + \frac{1}{2t} + O(t^{-2}) \Bigr) .
\]
\end{Corollary}

This corollary implies the first expansion in Lemma~\ref{lem:gammaz.muz.Sigmaz}.

\begin{proof}[\bf Proof of Corollary~\ref{cor:Gd.logGd}]
By means of \eqref{eq:expansion.hd} we can apply Lemma~\ref{lem:expansions} with $h = h_d$, $a = a_d$, $k = 2$ and
\begin{align*}
	(b_0,b_1,b_2) \
	&= \ 2^{a_d-1} C_d
		\Bigl( 1, - \frac{a_d - 1}{2}, \frac{[a_d - 1]_2}{8} \Bigr) \\
	&= \ \frac{2^{a_d-1} \Gamma(d/2)}{\Gamma(a_d)\Gamma(1/2)}
				\Bigl( 1, - \frac{a_d - 1}{2}, \frac{[a_d - 1]_2}{8} \Bigr) .
\end{align*}
Then Lemma~\ref{lem:expansions} implies that
\begin{align*}
	G_d(t) \
	&= \ \frac{2^{a_d - 1} \Gamma(d/2)}{\Gamma(a_d)\Gamma(1/2)} \, e_{}^t t_{}^{-a_d}
		\Bigl( 1 - \frac{\Gamma(a_d+1)(a_d - 1)}{2 t}
			+ \frac{\Gamma(a_d + 2) [a_d - 1]_2}{8 t^2} + O(t^{-3}) \Bigr) \\
	&= \ \frac{2^{a_d - 1} \Gamma(d/2)}{\Gamma(1/2)} \, e_{}^t t_{}^{-a_d}
		\Bigl( 1 - \frac{[a_d]_2}{2 t} + \frac{[a_d + 1]_4}{8 t^2} + O(t^{-3}) \Bigr)
\end{align*}
where we used the identities $\Gamma(a_d + 1)/\Gamma(a_d) = a_d$ and $\Gamma(a_d+2)/\Gamma(a_d) = [a_d+1]_2$. The previous and all subsequent expansions are meant as $t \to \infty$. Moreover,
\[
	\log G_d(t) \
	= \ \log \Bigl( \frac{2^{a_d-1} \Gamma(d/2)}{\Gamma(1/2)} \Bigr)
		+ t - a_d \log(t) + \log \Bigl( 1 - \frac{[a_d]_2}{2 t} + \frac{[a_d+1]_4}{8 t^2} + O(t^{-3} \Bigr) ,
\]
and the standard expansion $\log(1 + x) = x - x^2/2 + O(x^3)$ as $x \to 0$ implies that
\begin{align*}
	\log \Bigl( 1 - \frac{[a_d]_2}{2 t} + \frac{[a_d+1]_4}{8 t^2} + O(t^{-3}) \Bigr) \
	&= \ - \frac{[a_d]_2}{2 t} + \frac{[a_d+1]_4}{8 t^2} - \frac{[a_d]_2^2}{8 t^2} + O(t^{-3}) \\
	&= \ - \frac{[a_d]_2}{2 t} - \frac{[a_d]_2}{4 t^2} + O(t^{-3}) \\
	&= \ - \frac{[a_d]_2}{2 t} \Bigl( 1 + \frac{1}{2t} + O(t^{-2}) \Bigr) ,
\end{align*}
because $[a_d+1]_4 - [a_d]_2^2 = - 2 [a_d]_2$.
\end{proof}

Another consequence are expansions for moments of $U_t$. Precisely, we can express moments of $U_t$ in terms of moments of $1 - U_t$, and the latter may be approximated by means of Lemma~\ref{lem:expansions}.

\begin{Corollary}
\label{cor:moments.Ut}
As $t \to \infty$,
\begin{align*}
	\Ex(U_t) \
	&= \ 1 - \frac{a_d}{t} \Bigl( 1 - \frac{a_d-1}{2t} + O(t^{-2}) \Bigr) , \\
	\Var(U_t) \
	&= \ \frac{a_d}{t^2} \Bigl( 1 - \frac{a_d-1}{t} + O(t^{-2}) \Bigr) , \\
	\frac{1 - \Ex(U_t^2)}{d-1} \
	&= \ \frac{1}{t} \Bigl( 1 - \frac{a_d}{t} + O(t^{-2}) \Bigr) .
\end{align*}
\end{Corollary}

This corollary, applied to \eqref{eq:mu(z)} and \eqref{eq:Sigma(z)}, yields the second and third expansion in Lemma~\ref{lem:gammaz.muz.Sigmaz}.

\begin{proof}[\bf Proof of Corollary~\ref{cor:moments.Ut}]
We apply Lemma~\ref{lem:expansions} with $h = h_d$, $a = a_d$, $k = 1$ and $(b_0,b_1) = C (1, \tilde{b}_1)$ for some constant $C > 0$ and $\tilde{b}_1 = - (a_d-1)/2$. Note that Lemma~\ref{lem:expansions} implies that for any integer $\ell \ge 1$,
\begin{align*}
	\Ex[(1 - U_t)^\ell] \
	&= \ \int_{-1}^1 (1 - u)^\ell e^{tu} h_d(u) \, \d u \Bigr/
		\int_{-1}^1 e^{tu} h_d(u) \, \d u \\
	&= \ \frac{\Gamma(a_d+\ell)}{\Gamma(a_d) \, t^\ell} \,
		\frac{1 + \tilde{b}_1 (\ell+a_d)t^{-1} + O(t^{-2})}
			{1 + \tilde{b}_1 a_d t^{-1} + O(t^{-2})} \\
	&= \ \frac{[a_d + \ell - 1]_\ell}{t^\ell}
		\Bigl( 1 + \frac{\tilde{b}_1 \ell}{t} + O(t^{-2}) \Bigr) \\
	&= \ \frac{[a_d + \ell - 1]_\ell}{t^\ell}
		\Bigl( 1 - \frac{\ell(a_d-1)}{2t} + O(t^{-2}) \Bigr) ,
\end{align*}
Specifically, for $\ell = 1, 2$ we obtain the expansions
\begin{align*}
	\Ex(1 - U_t) \
	&= \ \frac{a_d}{t}
		\Bigl( 1 - \frac{a_d-1}{2t} + O(t^{-2}) \Bigr) , \\
	\Ex[(1 - U_t)^2] \
	&= \ \frac{[a_d + 1]_2}{t^2}
		\Bigl( 1 - \frac{a_d-1}{t} + O(t^{-2}) \Bigr) ,
\end{align*}
and this implies that
\begin{align*}
	\Var(U_t) \
	&= \ \Ex[(1 - U_t)^2] - [\Ex(1 - U_t)]^2 \\
	&= \ \frac{[a_d + 1]_2}{t^2}
		\Bigl( 1 - \frac{a_d-1}{t} + O(t^{-2}) \Bigr)
			- \frac{a_d^2}{t^2} \Bigl( 1 - \frac{a_d - 1}{2t} + O(t^{-2}) \Bigr)^2 \\
	&= \ \frac{[a_d + 1]_2}{t^2}
		\Bigl( 1 - \frac{a_d-1}{t} + O(t^{-2}) \Bigr)
			- \frac{a_d^2}{t^2} \Bigl( 1 - \frac{a_d - 1}{t} + O(t^{-2}) \Bigr) \\
	&= \ \frac{a_d}{t^2}
		\Bigl( 1 - \frac{a_d-1}{t} + O(t^{-2}) \Bigr) , \\
	1 - \Ex(U_t^2) \
	&= \ 2 \Ex(1 - U_t) - \Ex[(1 - U_t)^2] \\
	&= \ \frac{2 a_d}{t} - \frac{[a_d]_2}{t^2} - \frac{[a_d + 1]_2}{t^2} + O(t^{-3}) \\
	&= \ \frac{2 a_d}{t} - \frac{2 a_d^2}{t^2} + O(t^{-3}) \\
	&= \ \frac{d-1}{t} \Bigl( 1 - \frac{a_d}{t} + O(t^{-2}) \Bigr) ,
\end{align*}
because $2a_d = d-1$.
\end{proof}

\section{Derivatives of the negative log-likelihood functions}
\label{app:Derivatives.weighted.NLL}

In the regression settings of Section~\ref{sec:Regression}, suppose that each function $\bsf \in \FF$ is equal to $\bsf = (f_k)_{k=1}^d$ with all components $f_k$ belonging to the same finite-dimensional linear space $\FF^o$ of functions $f^o : \XX \to \R$. If $f^o_1,\ldots,f^o_r$ is a basis of $\FF^o$, then any function $\bsf \in \FF$ can be represented as
\[
	\bsf(\x) \ = \ \bTheta \bF^o(\x)
\]
with a parameter matrix $\bTheta \in \R^{d\times r}$ and $\bF^o(\x) := (f^o_j(\x))_{j=1}^r \in \R^r$. This leads to the negative log-likelihood function $L : \R^{d\times r} \to \R$,
\[
	L(\bTheta) \ := \ \ell(\bTheta \bF^o) \
	= \ \sum_{i=1}^n W_i \, \bigl( \gamma(\bTheta\bF^o(\Xi)) - \Yi^\top \bTheta\bF^o(\Xi) \bigr) .
\]
Here $W_i = 1$ for a parametric GLM, and $W_i = w_{\x_o}(X_i)$ for the local GLMs. Since $\bmu(\z)$ and $\bSigma(\z)$ are the gradient and Hessian matrix of $\gamma$ at $\z$, for $\bTheta, \bDelta \in \R^{d\times r}$,
\begin{align*}
	L(\bTheta + \bDelta) \
	= \ &L(\bTheta) + \sum_{i=1}^n W_i \, (\bmu(\bTheta \bF^o(\Xi)) - \Yi)^\top \bDelta \Xi \\
		&+ \ \frac{1}{2} \sum_{i=1}^n W_i \,
			\bF^o(\Xi)^\top \bDelta^\top \bSigma(\bTheta\Xi) \bDelta \bF^o(\Xi)
		+ O(\|\bDelta\|_F^3)
\end{align*}
as $\bDelta \to 0$. Here $\|\bs{A}\|_F$ is the Frobenius norm of a matrix $\bs{A}$. If we define $\langle \bs{A},\bs{B}\rangle = \trace(\bs{A}^\top \bs{B})$ for matrices $\bs{A}, \bs{B}$ of the same size, then $\|\bs{A}\|_F = \langle \bs{A},\bs{A}\rangle^{1/2}$, and
\[
	\sum_{i=1}^n W_i \, (\bmu(\bTheta \bF^o(\Xi)) - \Yi)^\top \bDelta \bF^o(\Xi) \
	= \ \langle \bDelta, \bs{G}(\bTheta) \rangle
\]
with the gradient matrix
\begin{equation}
\label{eq:Gradient.LTheta}
	\bs{G}(\bTheta) \ := \ \sum_{i=1}^n W_i \,
		(\bmu(\bTheta \bF^o(\Xi)) - \Yi) \Xi^\top \ \in \ \R^{d\times r} .
\end{equation}
Moreover, if $\bDelta = [\bDelta_1,\ldots,\bDelta_r]$ with $\bDelta_j \in \R^d$, then for $\x \in \XX$ and $\bSigma \in \R^{d\times d}_{\rm sym}$,
\[
	\bF^o(\x)^\top \bDelta^\top \bSigma \bDelta \bF^o(\x) \
	= \ \mathrm{vec}(\bDelta)^\top \bigl( \bF^o(\x)\bF^o(\x)^\top \otimes \bSigma) \, \mathrm{vec}(\bDelta) ,
\]
where
\begin{align*}
	\mathrm{vec}(\bDelta) \
	&:= \ [\bDelta_1^\top, \bDelta_2^\top, \ldots, \bDelta_r^\top]^\top \
		\in \ \R^{dr} , \\
	\bF^o(\x)\bF^o(\x)^\top \otimes \Sigma \
	&:= \ \begin{bmatrix}
			f_1^o(\x) f_1^o(\x) \bSigma & f_1^o(\x) f_2^o(\x) \bSigma & \ldots & f_1^o(\x) f_r^o(\x) \bSigma \\
			f_2^o(\x) f_1^o(\x) \bSigma & f_2^o(\x) f_2^o(\x) \bSigma & \ldots & f_2^o(\x) f_r^o(\x) \bSigma \\
			\vdots        & \vdots        & \ddots & \vdots \\
			f_r^o(\x) f_1^o(\x) \bSigma & f_r^o(\x) f_2^o(\x) \bSigma & \ldots & f_r^o(\x) f_r^o(\x) \bSigma
		\end{bmatrix} \
		\in \ \R^{dr \times dr}_{\rm sym}
\end{align*}
Hence, the Hessian matrix for $L(\bTheta)$, viewed as a function of $\mathrm{vec}(\bTheta)$, equals
\begin{equation}
\label{eq:Hessian.LTheta}
	\sum_{i=1}^n W_i \, \bF^o(\Xi) \bF^o(\Xi)^\top \otimes \bSigma(\bTheta \Xi) .
\end{equation}
These formulae are useful to minimize $L(\cdot)$ over $\R^{d\times r}$ via a Newton--Raphson procedure.

\section{Further details about the simulation study}
\label{app:Further.details.numerical.example}

Figure~\ref{fig:GLM_local.weights} illustrates the weights $w_{\x_o}(\Xi)$ for two different choices of $\x_o$ and the weights defined via \eqref{eq:local.weights} with three different choices for $N$. The weights $w_{\x_o}(\X_i)$ are coded on a gray scale with $1$ corresponding to $1$ and almost white corresponding to $0$. The coordinates of $\x_o$ are indicated by auxiliary red lines.

\begin{figure}
\includegraphics[width=0.47\textwidth]{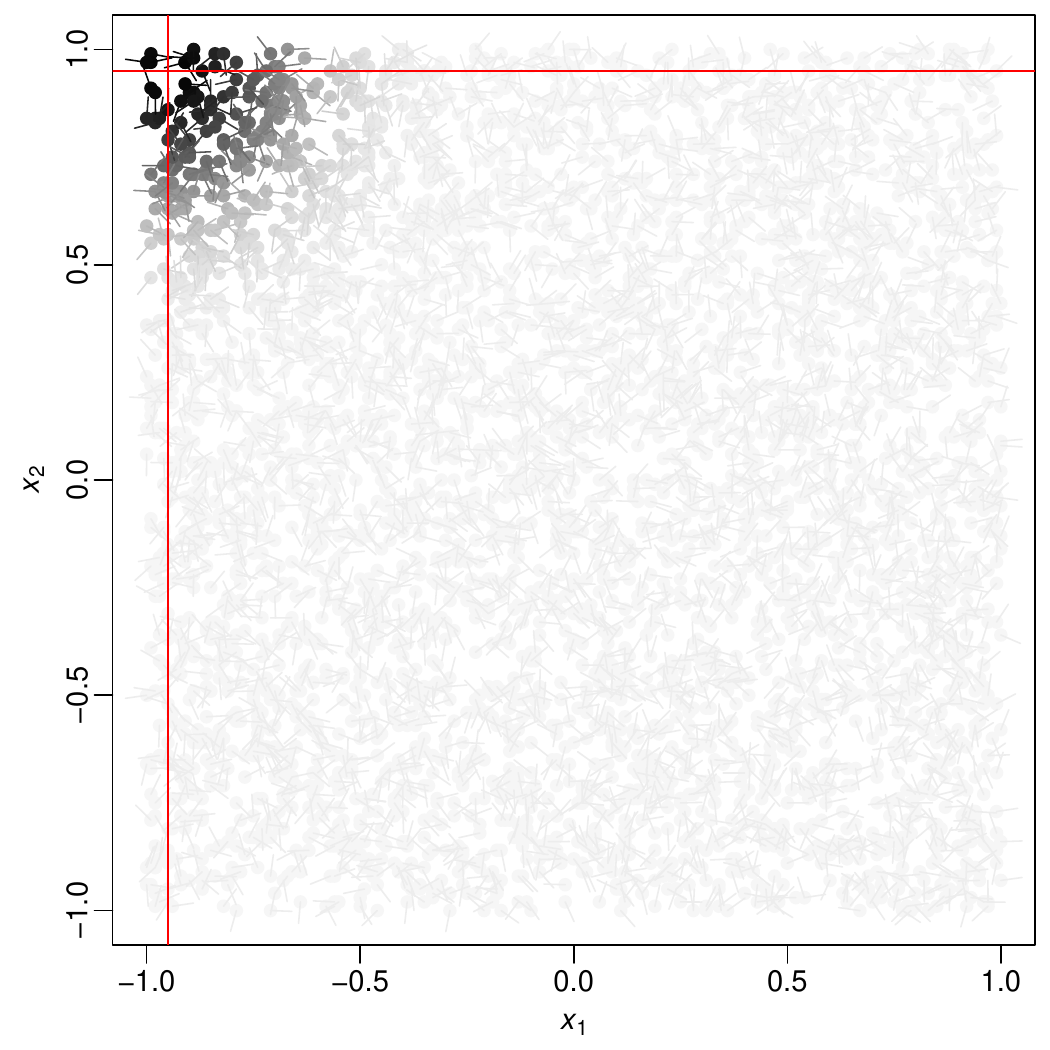}
\hfill
\includegraphics[width=0.47\textwidth]{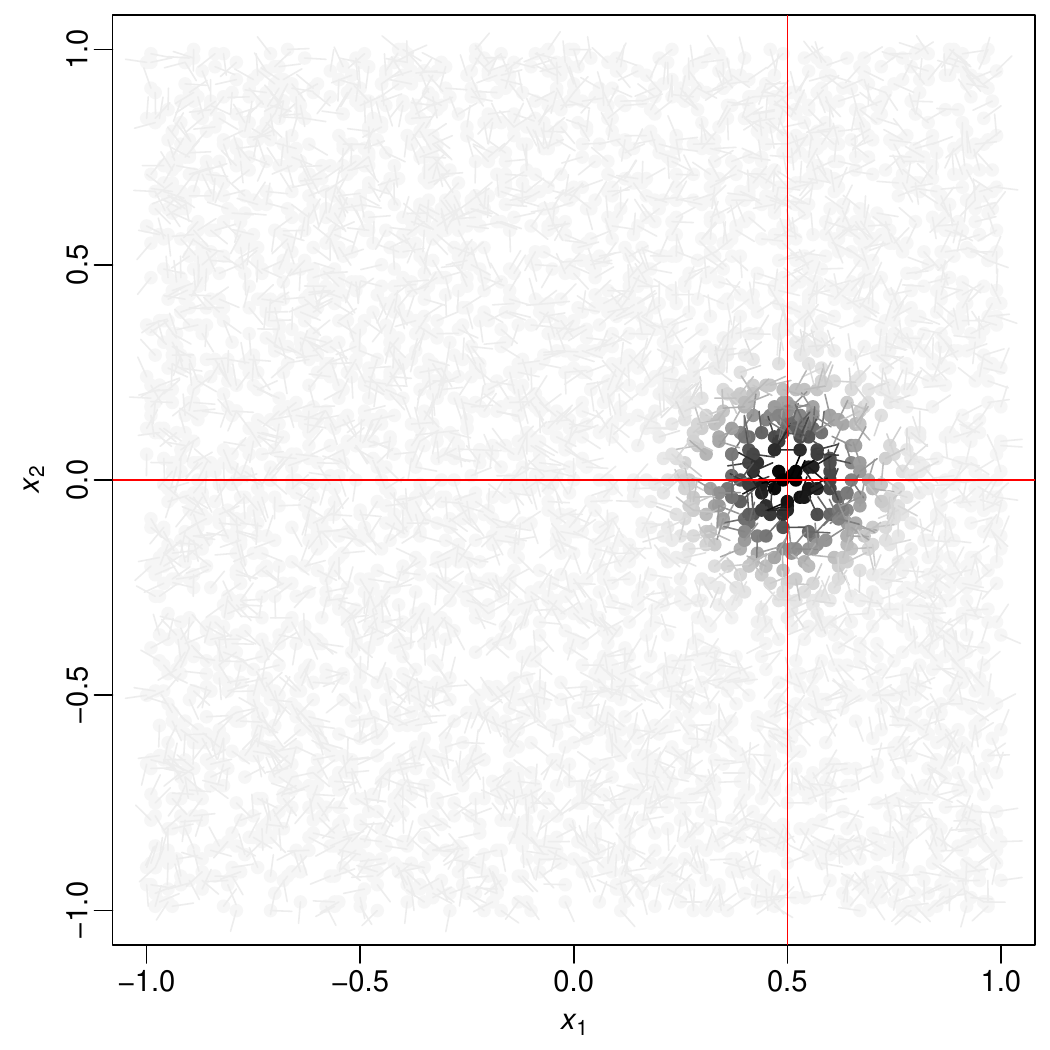}

\includegraphics[width=0.47\textwidth]{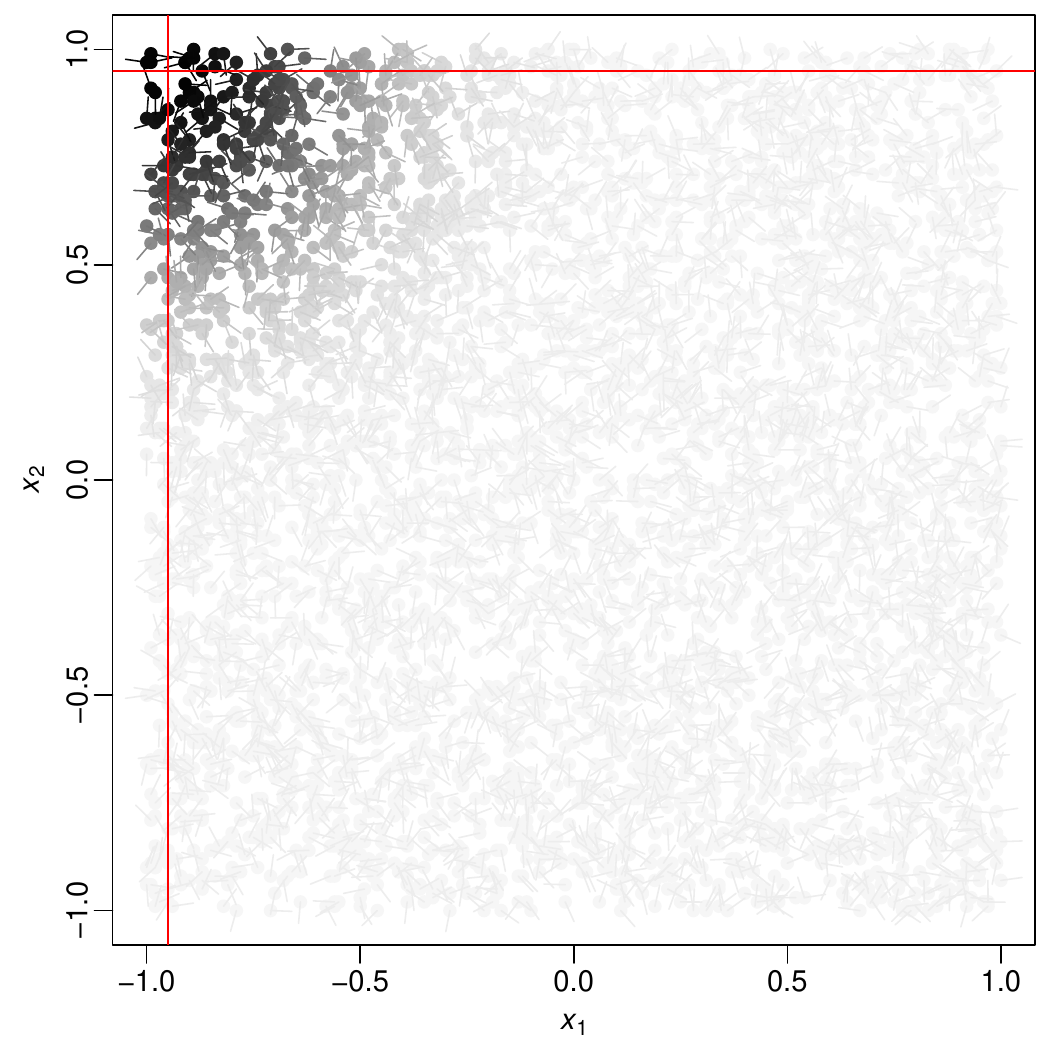}
\hfill
\includegraphics[width=0.47\textwidth]{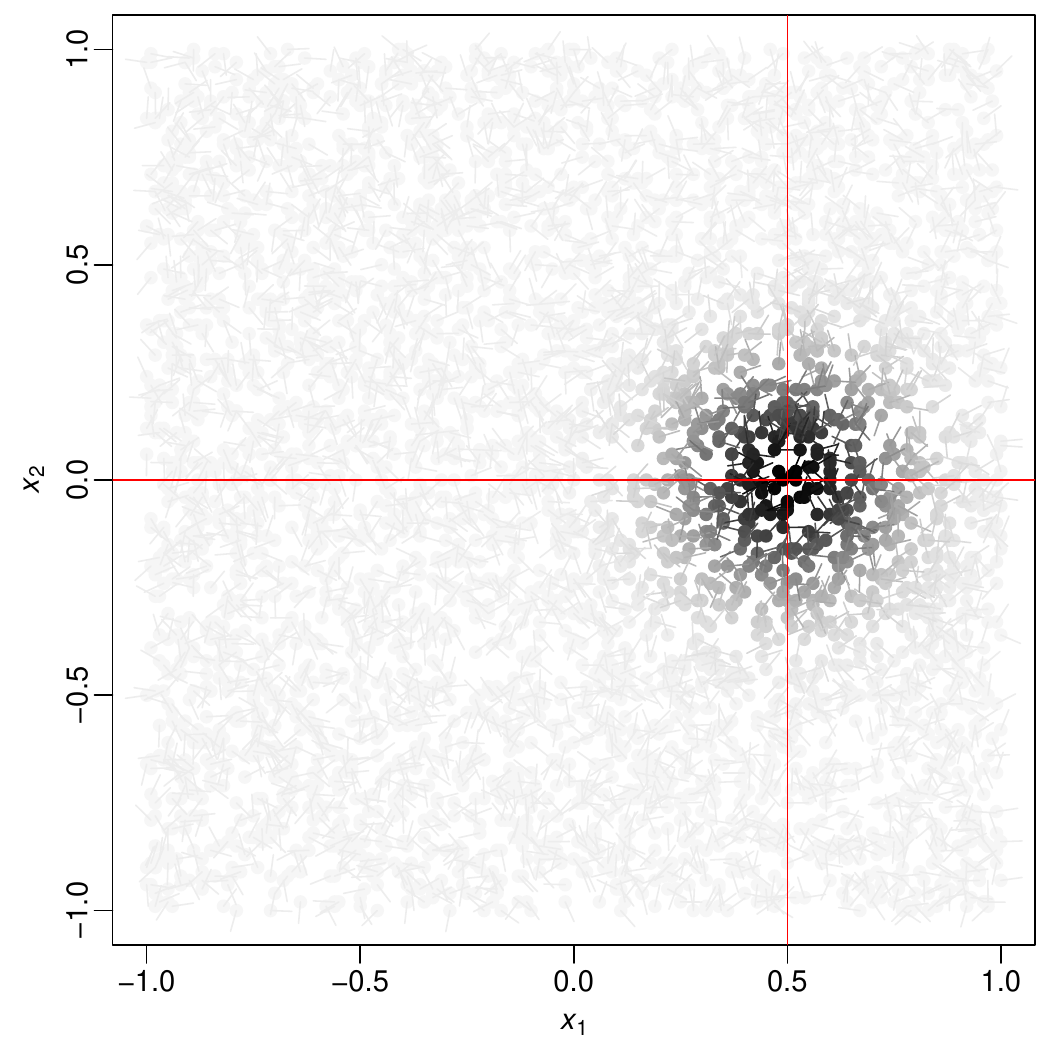}

\includegraphics[width=0.47\textwidth]{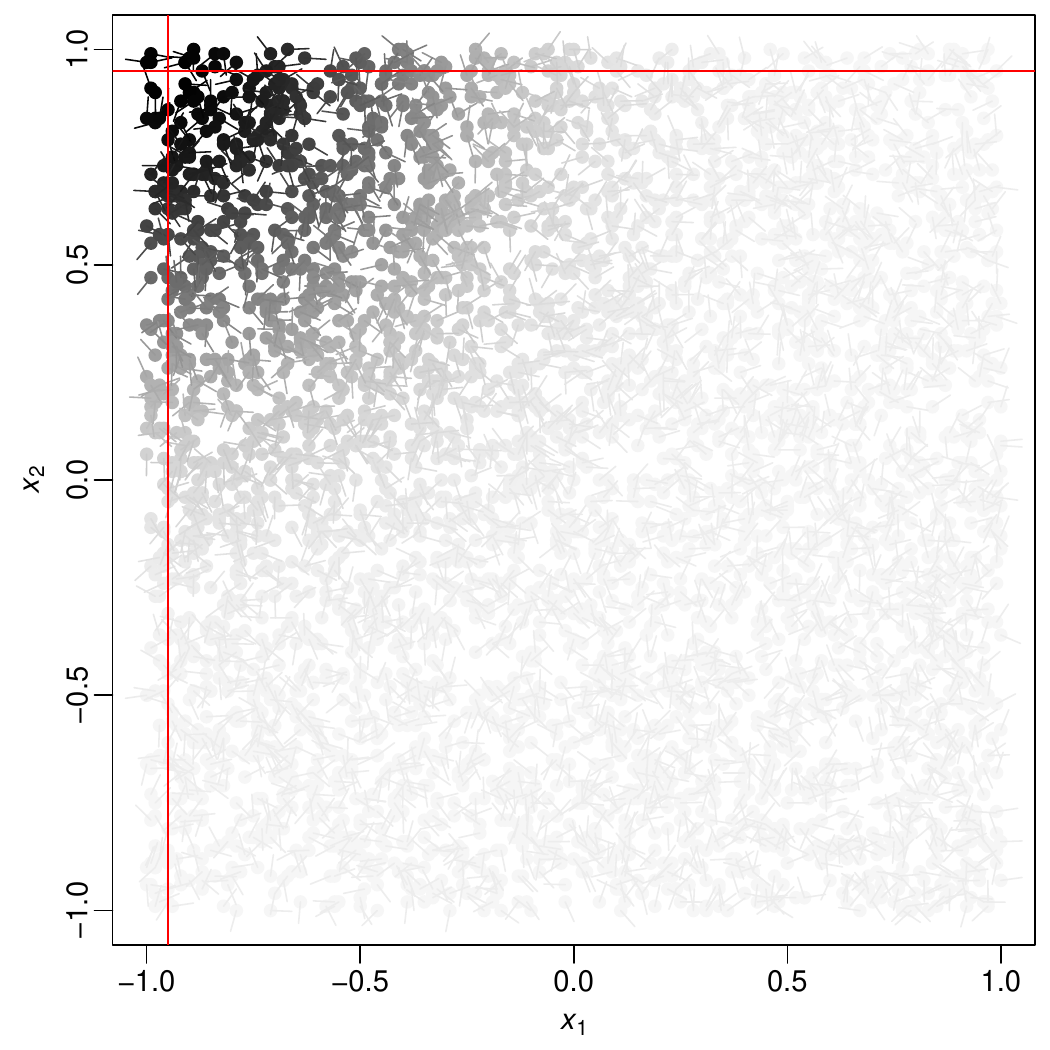}
\hfill
\includegraphics[width=0.47\textwidth]{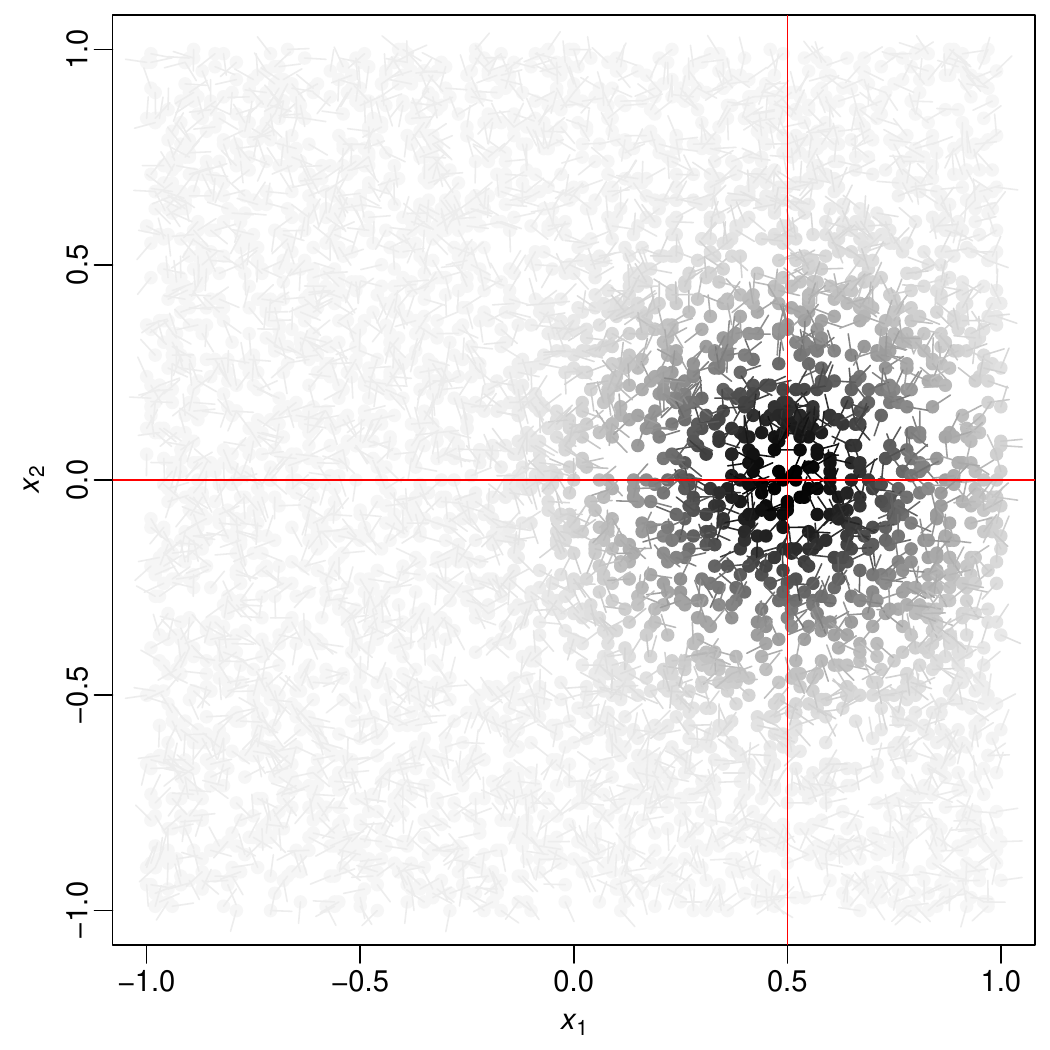}

\caption{Locally weighted raw data via \eqref{eq:local.weights} with $\x_o = [-0.95,0.95]^\top$ (left column) and $\x_o = [0.5,0]^\top$ (right column), where $N = 100$ (1st row), $N = 200$ (2nd row) or $N = 400$ (3rd row).}
\label{fig:GLM_local.weights}
\end{figure}

Figure~\ref{fig:GLM_bias} depicts the bias $\mathrm{Bias}(\x_o) := \Ex \bmu(\hat{\bsf}(\x_o)) - \bmu(\bsf^*(\x_o))$ for six different estimators at each point $\x_o \in \XX_o$ by a line segment connecting $\x_o$ (gray bullet) with $\x_o + 0.7 \cdot \mathrm{Bias}(\x_o)$. One sees clearly that local quadratic modelling invokes a relatively small bias, except in the corners of the domain.

\begin{figure}
\includegraphics[width=0.47\textwidth]{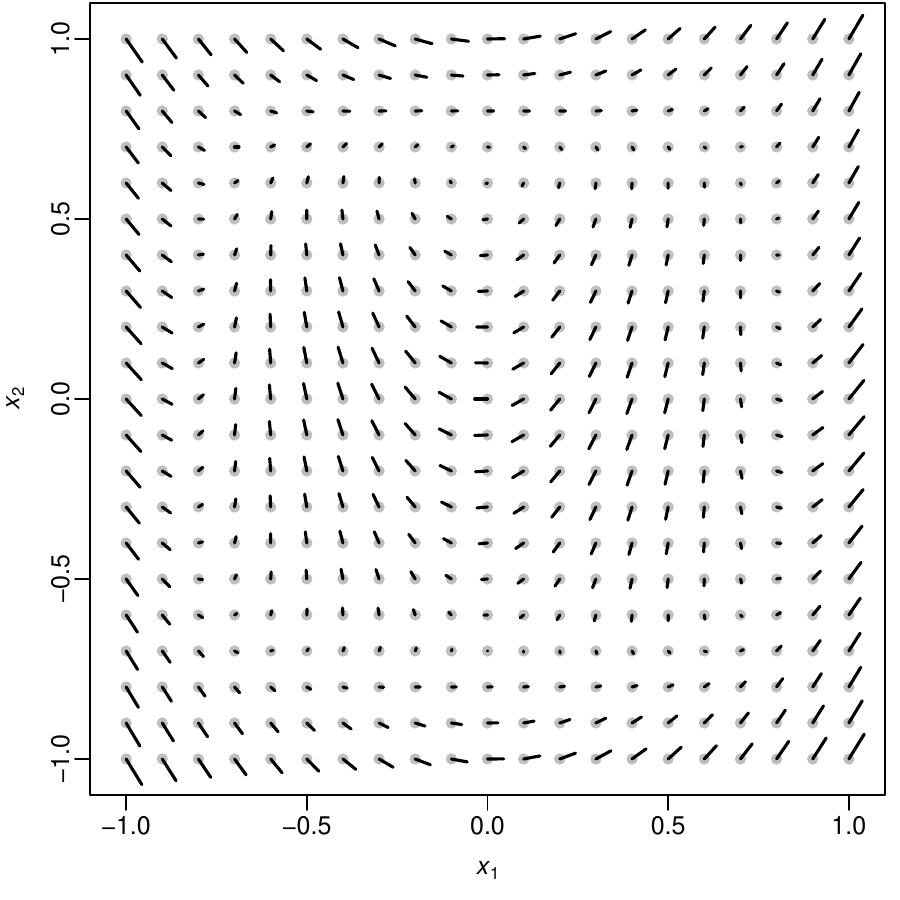}
\hfill
\includegraphics[width=0.47\textwidth]{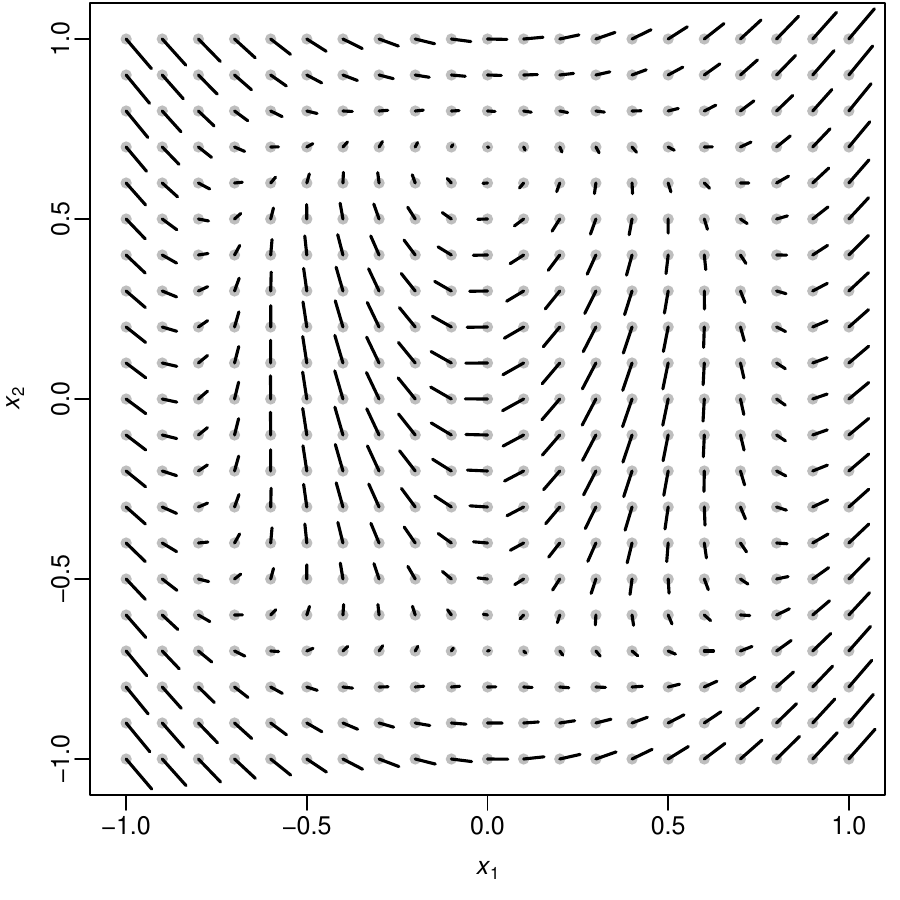}

\includegraphics[width=0.47\textwidth]{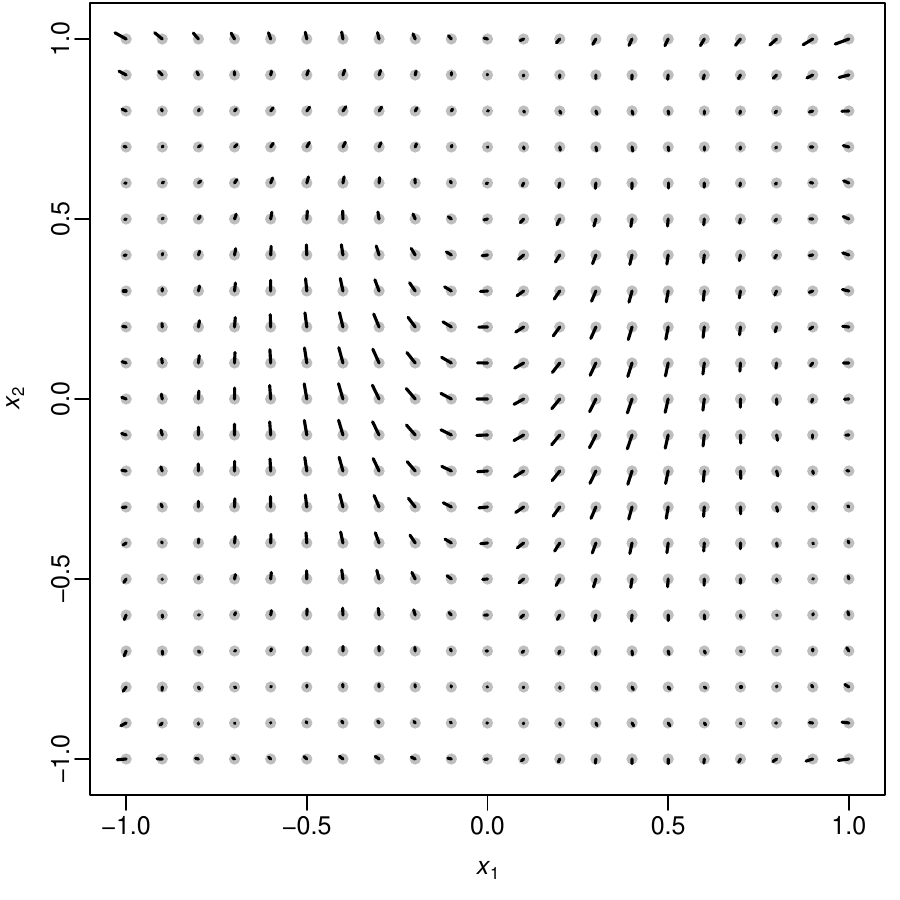}
\hfill
\includegraphics[width=0.47\textwidth]{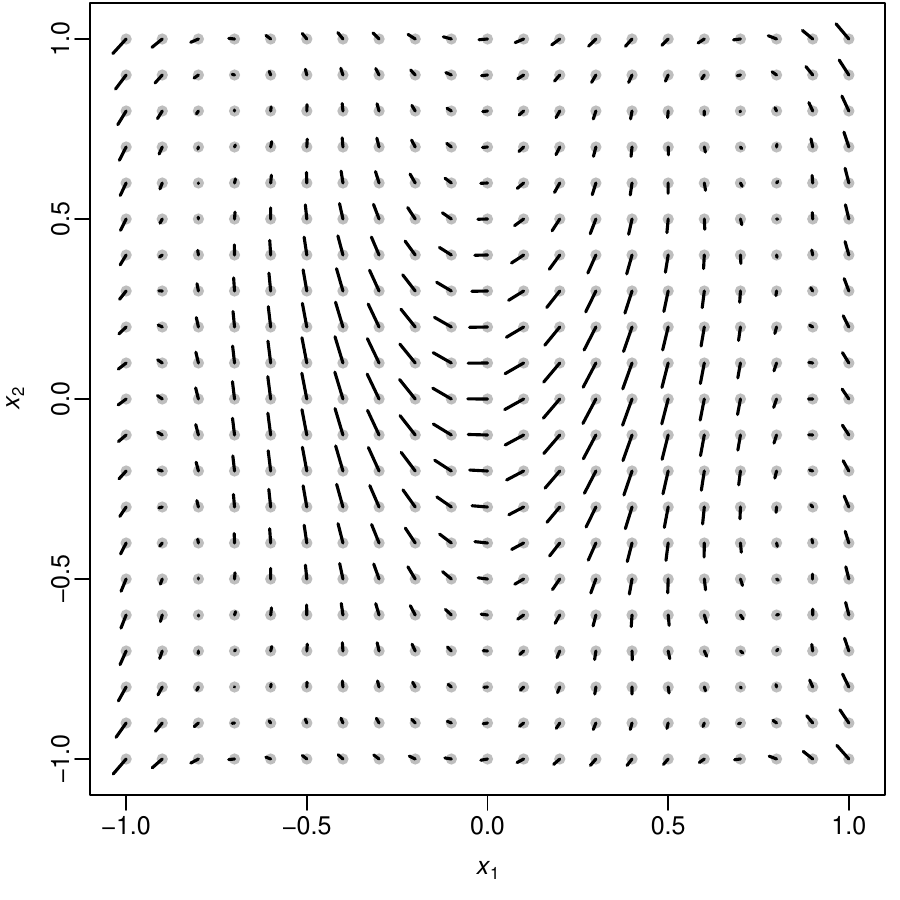}

\includegraphics[width=0.47\textwidth]{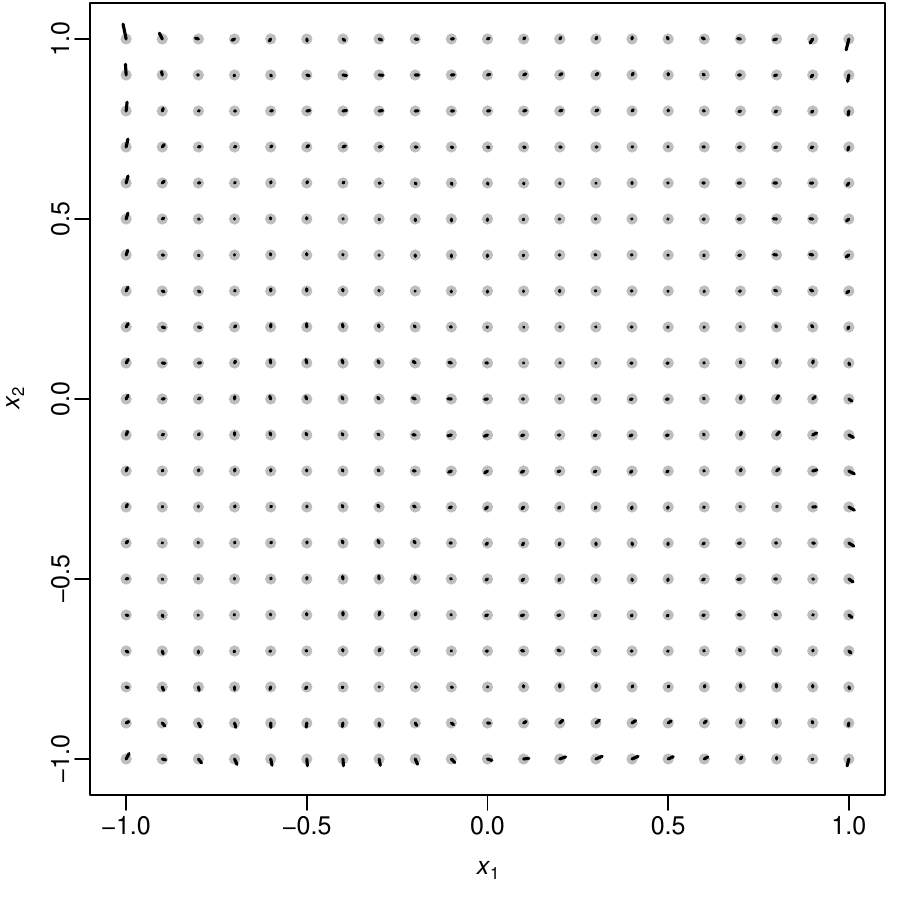}
\hfill
\includegraphics[width=0.47\textwidth]{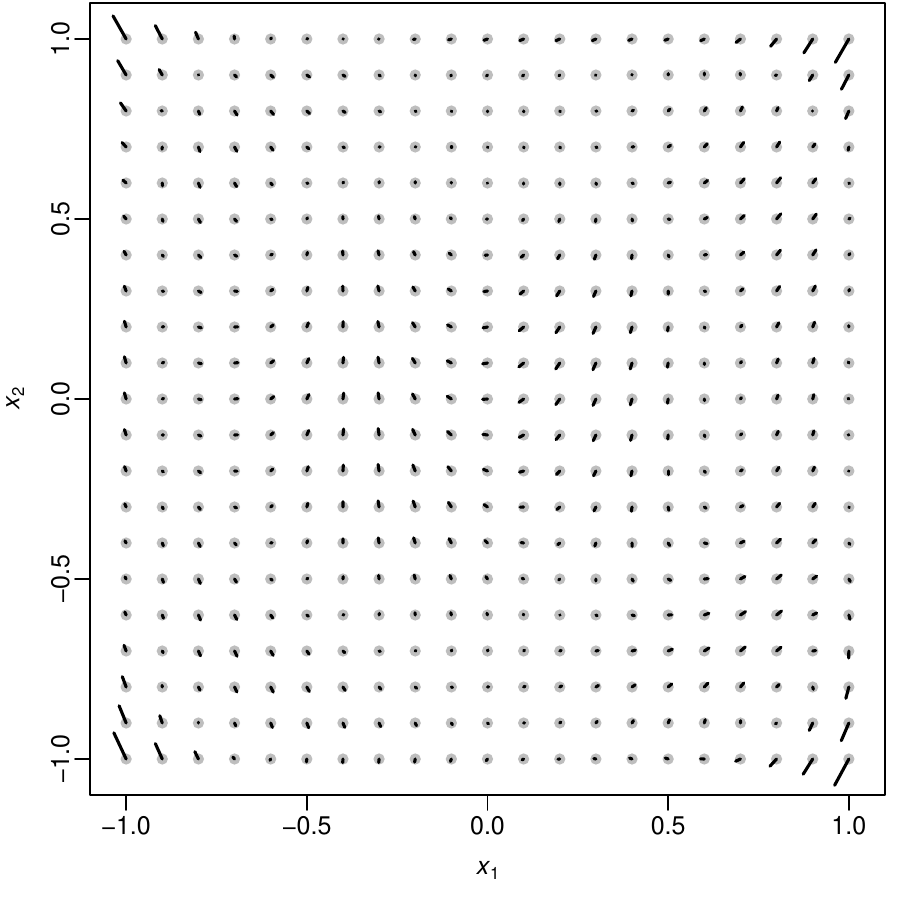}

\caption{Bias function for the local constant (1st row), linear (2nd row) and quadratic (3rd row) estimators, where $N = 200$ (left column) or $N = 400$ (right column).}
\label{fig:GLM_bias}
\end{figure}

\paragraph{Acknowledgements.}
The authors are grateful to Riccardo Gatto and Rudy Beran for stimulating discussions and useful references about directional data. Many thanks also to Alyssa Rhoden and Costanza Rossi for sharing data and insights about Europa and Ganymede, respectively, and their interest in the methods presented here. Constructive comments of an anonymous reviewer led to further improvements.

C.\ Haslebacher acknowledges financial support of the Swiss National Science Foundation (SNSF) through grants P500PT\_225447 and 51NF40\_205606. L.\ D\"umbgen acknowledges the support of SNSF through grant 10001553.


\begin{thebibliography}{15}
\expandafter\ifx\csname natexlab\endcsname\relax\def\natexlab#1{#1}\fi
\expandafter\ifx\csname url\endcsname\relax
  \def\url#1{\texttt{#1}}\fi
\expandafter\ifx\csname urlprefix\endcsname\relax\def\urlprefix{URL }\fi

\bibitem[{Arnold and Jupp(2013)}]{Arnold_Jupp_2013}
\textsc{Arnold, R.} and \textsc{Jupp, P.~E.} (2013).
\newblock Statistics of orthogonal axial frames.
\newblock \textit{Biometrika} \textbf{100} 571--586.

\bibitem[{Barndorff-Nielsen(2014)}]{Barndorff-Nielsen_2014}
\textsc{Barndorff-Nielsen, O.} (2014).
\newblock \textit{Information and exponential families in statistical theory}.
\newblock Wiley Series in Probability and Statistics, John Wiley \& Sons, Ltd.,
  Chichester.

\bibitem[{Bingham(1974)}]{Bingham_1974}
\textsc{Bingham, C.} (1974).
\newblock An antipodally symmetric distribution on the sphere.
\newblock \textit{Ann. Statist.} \textbf{2} 1201--1225.

\bibitem[{Fan et~al.(1998)Fan, Farmen and Gijbels}]{Fan_etal_1998}
\textsc{Fan, J.}, \textsc{Farmen, M.} and \textsc{Gijbels, I.} (1998).
\newblock Local maximum likelihood estimation and inference.
\newblock \textit{J. R. Stat. Soc. Ser. B Stat. Methodol.} \textbf{60}
  591--608.

\bibitem[{Haslebacher et~al.(2025)Haslebacher, Tejero, Prockter, Leonard,
  Rhoden and Thomas}]{HaslebacherPSJ2025}
\textsc{Haslebacher, C.}, \textsc{Tejero, J.~G.}, \textsc{Prockter, L.~M.},
  \textsc{Leonard, E.~J.}, \textsc{Rhoden, A.~R.} and \textsc{Thomas, N.}
  (2025).
\newblock Length, width, and relative age analysis of lineaments in the
  {G}alileo regional maps with lineamapper.
\newblock \textit{PSJ}  under review.

\bibitem[{Haslebacher et~al.(2024)Haslebacher, Thomas and
  Bickel}]{Haslebacher2024a}
\textsc{Haslebacher, C.}, \textsc{Thomas, N.} and \textsc{Bickel, V.~T.}
  (2024).
\newblock Lineamapper: A deep learning-powered tool for mapping linear surface
  features on {E}uropa.
\newblock \textit{Icarus} \textbf{410}.

\bibitem[{Leonard et~al.(2024)Leonard, Patthoff and Senske}]{Leonard2024}
\textsc{Leonard, E.~J.}, \textsc{Patthoff, A.~D.} and \textsc{Senske, D.~A.}
  (2024).
\newblock Global geologic map of {E}uropa.
\newblock \textit{USGS/NASA}  3513.

\bibitem[{Mardia and Jupp(2000)}]{Mardia_Jupp_2000}
\textsc{Mardia, K.~V.} and \textsc{Jupp, P.~E.} (2000).
\newblock \textit{Directional statistics}.
\newblock Wiley Series in Probability and Statistics, John Wiley \& Sons, Ltd.,
  Chichester.

\bibitem[{McCullagh and Nelder(1989)}]{McCullagh_Nelder_1989}
\textsc{McCullagh, P.} and \textsc{Nelder, J.~A.} (1989).
\newblock \textit{Generalized linear models}.
\newblock 2nd ed. Monographs on Statistics and Applied Probability, Chapman \&
  Hall, London.

\bibitem[{Pewsey and Garc\'ia-Portugu\'es(2021)}]{Pewsey_Garcia-Portugues_2021}
\textsc{Pewsey, A.} and \textsc{Garc\'ia-Portugu\'es, E.} (2021).
\newblock Recent advances in directional statistics.
\newblock \textit{TEST} \textbf{30} 1--58.

\bibitem[{{R Core Team}(2023)}]{R2023}
\textsc{{R Core Team}} (2023).
\newblock \textit{R: A Language and Environment for Statistical Computing}.
\newblock R Foundation for Statistical Computing, Vienna, Austria.
\newline\urlprefix\url{https://www.R-project.org/}

\bibitem[{Rhoden and Hurford(2013)}]{Rhoden2013}
\textsc{Rhoden, A.~R.} and \textsc{Hurford, T.~A.} (2013).
\newblock Lineament azimuths on {E}uropa: Implications for obliquity and
  non-synchronous rotation.
\newblock \textit{Icarus} \textbf{226} 841--859.

\bibitem[{Rossi et~al.(2020)Rossi, Cianfarra and Salvini}]{Rossi2020}
\textsc{Rossi, C.}, \textsc{Cianfarra, P.} and \textsc{Salvini, F.} (2020).
\newblock Structural geology of ganymede regional groove systems (60°n-60°s).
\newblock \textit{Journal of Maps} \textbf{16} 6--16.

\bibitem[{Sabbeth et~al.(2023)Sabbeth, Smrekar and Stock}]{Sabbeth2023}
\textsc{Sabbeth, L.}, \textsc{Smrekar, S.~E.} and \textsc{Stock, J.~M.} (2023).
\newblock Estimated seismicity of venusian wrinkle ridges based on fault
  scaling relationships.
\newblock \textit{Earth and Planetary Science Letters} \textbf{619} 118308.

\bibitem[{Watters et~al.(2001)Watters, Cook and Robinson}]{Watters2001}
\textsc{Watters, T.~R.}, \textsc{Cook, A.~C.} and \textsc{Robinson, M.~S.}
  (2001).
\newblock Large-scale lobate scarps in the southern hemisphere of mercury.
\newblock \textit{Planetary and Space Science} \textbf{49} 1523--1530.

\end{thebibliography}

\end{document}